\documentclass[12pt,reqno]{amsart}
\usepackage{graphicx,url}
\usepackage{caption}
\usepackage{subcaption}
\usepackage{setspace}
\usepackage{empheq}
\usepackage{cite}
\usepackage{amsmath}
\usepackage{mathtools}
\usepackage{textcomp}
\usepackage{amssymb}
\usepackage{tikz-cd}
\usepackage[normalem]{ulem}
\usepackage{comment}

\usepackage{multirow}
\usepackage{color}
\usepackage{xcolor}
\definecolor{MyLinkColor}{rgb}{0,0,0.4}

\newcommand{\tb}{\textcolor{blue}}
\newcommand{\tcr}{\textcolor{red}}

\newcommand{\vp}{\varphi}
    \DeclareMathOperator{\csch}{csch}

\newtheorem{thm}{Theorem}[section]
\newtheorem{prop}[thm]{Proposition}

\theoremstyle{remark} 
\newtheorem{rem}[thm]{Remark}

\numberwithin{equation}{section}   

\title[Three-layer water flows]{Three-layer water flows: Dirichlet-Neumann operators and approximations}
\author[R. I. Ivanov ]{Rossen Ivanov }
\address{School of Mathematics and Statistics, Technological University Dublin, City Campus, Grangegorman Lower, Dublin D07 ADY7, Ireland.}
\email{rossen.ivanov@tudublin.ie}
\author[C. I. Martin]{Calin-Iulian Martin }
\address{Faculty of Mathematics and Computer Science, Babes-Bolyai University, Mihail Kogalniceanu Str. 1, Cluj-Napoca, Romania.}
\email{calin.martin@ubbcluj.ro}

\subjclass[2010]{35Q31, 35Q35.}
\keywords{Nonlinear water waves, discontinuous density stratification, Hamiltonian formulation, Dirichlet-Neumann operators, root separation}
\begin{document}
 \begin{abstract}
The object of investigation in this paper are the nonlinear equations of motion for two-dimensional inviscid water flows with piecewise constant density stratification in a three-layer fluid with a flat bottom, a free surface and two interfaces. We establish a Hamiltonian formulation for the nonlinear governing equations in this setup.
The Hamiltonian of the system and the equations of motion of the surface and of the interfaces are expressed with the help of the Dirichlet-Neumann (DN) operators, which are introduced for each of the layers. 
Then, the linear equations for small amplitudes of the elevation of the surface and of the interfaces in the leading order are derived from which a bi-cubic equation for the dispersion relation is obtained, whose solutions are analysed. The six real solutions for the possible propagation speeds (three positive, related to right-moving waves and three negative, related to left-moving waves) have magnitudes of different order.
Upper and lower bounds for the previously mentioned roots are also given in terms of the coefficients of the equation. Subsequently, approximate formulae for the propagation speeds are derived. 
The importance of the DN operators is further illustrated in a separate analysis of the three-layer model with flat surface (rigid lid). The full nonlinear evolution equations are expressed again in terms of the DN operators, and the equations in the linear regime and the weakly nonlinear propagation regime (the Boussinesq approximation) are derived by a proper expansion of the DN operators. Limits to the two-layer free surface model are obtained as well. {The obtained results are applicable to internal waves in lakes and in the ocean as well as to laboratory experiments with three superimposed fluid layers. }
\end{abstract}

\maketitle
\section{Introduction}
 {Flow stratification—assumed to be of piecewise constant type in our work—is an intrinsic characteristic of geophysical water flows, particularly relevant to large-scale oceanic movements (cf. \cite{CJPOF19, Gill, MP, Val}). Strongly influenced by variations in temperature and salinity, density fluctuations occur frequently in the ocean and lead to a layered structure: fluid parcels of differing densities arrange themselves such that denser layers lie beneath lighter ones (cf. \cite{BecCushRoi, CJGAFD, Fed, Kes}). This vertical stratification results in the formation of interfaces between layers of distinct densities.}
Such a pronounced density stratification occurs within a region approximately 150 km wide on each side of the Equator, extending longitudinally over about 16,000 km in the Pacific Ocean. {This stratification gives rise to a sharp thermocline (also referred to as a pycnocline), which acts as a waveguide for internal waves}, cf. Fedorov \& Brown \cite{Fed}, and Constantin \& Johnson \cite{CJGAFD}. At times, internal waves might have appreciable amplitudes: field data provides evidence for tide-generated internal waves with amplitudes in excess of 100 m, extending over 100 km, cf. \cite{Sim}. Detailed, exact solutions depicting density stratified geophysical flows exhibiting internal waves were presented by Constantin \cite{ConsJPO13, ConsJPO14}.

From a practical side, the understanding of internal waves dynamics is important for mixing processes in oceans, submarine navigation as well as for the need to quantify induced loads on submerged marine platforms.

{Concerning the vertical stratification structure, it is widely accepted in oceanography (see \cite{Suth}) that the ocean is effectively subdivided into layers according to the strength of its stratification. Accordingly we have:  
 adjacent to the surface lies the \emph{mixed layer}, named for the wind-driven stirring of surface waters; beneath it is the the strongly stratified \emph{thermocline}, and below that, the weakly stratified \emph{deep layer} composed of cold, dense water that constitutes the bulk of the ocean’s mass. Complementing this conceptual model, a series of laboratory experiments using a wind flume \cite{Shi} simulated a three-layer density-stratified water body in which the \emph{middle layer (thermocline)} was intentionally made thick, closely matching the thickness of the upper and lower layers. This deliberate configuration was designed to reflect stratification patterns observed in natural lakes and to isolate the role of the thermocline. The study revealed that the thickness of the middle layer is a critical factor in shaping the system’s response to wind forcing, significantly influencing internal wave dynamics, energy transfer, and vertical mixing across the stratified column.}

Minding the previous aspects concerning the type of layering in physically reasonable ocean flows, we consider in this study three-layer stratified water flows: the fluid domain is bounded above by the free surface, below by the (rigid) bottom boundary and is split in three layers of constant, but different density. This type of stratification significantly complicates the analysis of an already highly intricate nonlinear problem. To alleviate these difficulties we will rewrite the governing equations by means of a Hamiltonian (re)formulation. Such a reformulation highlights structural properties and has the advantage that it also reduces the numbers of variables in a way that results in significant simplifications.
Note that the Hamiltonian formulation of the governing equations for two dimensional  gravity water flows was pioneered in \cite{Z68} for irrotational flows, and was extended to rotational flows with constant vorticity in \cite{CIP, Wah1, Wah2, CompelliIvanov2}, variable bottom \cite{CGNS,CIT,CIMT} etc. The Hamiltonian formulation for two-dimensional irrotational two-layer gravity water flows with a free surface was developed in \cite{CGK}, and the rotational counterpart with constant vorticity and constant density in each layer was obtained in \cite{CIM} for the case of periodic flows. The validity of the Hamiltonian perspective in the presence of Coriolis forces in the equatorial $f$-plane approximation, for zonal flows representing localized perturbations of an underlying pure current background state was established in \cite{CICMP}.
We would like to note that the power and relevance of Hamiltonian methods in the study of water flows is reinforced by the work of Craig, Guyenne \& Sulem \cite{CGS2} where a relevant extension to three-dimensional flows is presented. Additionally, a mix of Hamiltonian methods and techniques related to the Dirichlet-Neumann operators is instrumental in the derivation of coupled nonlinear equations relating the interfaces and the traces of the velocity potentials on them, in various propagation regimes, cf. e.g \cite{CGK, CGS2} and Section \ref{Boussinesq} of this work.

The two- and three- layer dynamics has been studied actively and various approximations have been employed, 
from \cite{Kaku} and \cite{Ovs} to the more recent \cite{Cam1,Cam2,Cam3,CHP,OH,Grim,Mil,Lan08}. 
For recent results concerning exact solutions (in Lagrangian and Eulerian coordinates) describing water flows with stratification of discontinuous type we refer the reader to \cite{ConsJPO13, ConsJPO14, Esc1, MarPoF21, Mat, Mas}.

Our aim is to present a Hamiltonian formulation and the full nonlinear equations of motion of two-dimensional stratified  water flows consisting of three layers, each of which having constant density. 
We adopt a versatile framework for our investigations: the densities in each layer can have any values, with the proviso that the stratification is of stable type. Moreover, the layers can have any thicknesses, without any constraints. 

The paper is organised as follows. After introducing the model setup and the governing equations in Section \ref{model_eq}, we show in Section \ref{reform_DN} that they can be rewritten as an infinite-dimensional Hamiltonian system. 
Furthermore, in Section \ref{reform_DN} we present the resulting full nonlinear equations of motion. They are expressed in terms of the so-called Dirichlet-Neumann Operators, which are introduced for the three-layer setting with a free surface. 
In Section \ref{disp} the dispersion relation for the leading order of the wave amplitudes is derived for arbitrary wavelengths. Section \ref{long} is devoted to the dispersion relation in the long-wave limit which is an intricate bi-cubic equation, whose solutions are all real numbers according to arguments presented in
 \cite{Ben, Bai, Yih}. Exploiting a result by Prodanov \cite{Prod, Prod2}, we  give lower and upper bounds for these solutions. The limit to two-layers is presented in Section \ref{2layers}. Approximate solutions in certain regimes are also given. Section \ref{rigid} is devoted to the three-layer rigid lid case. There again we present the full equations in terms of the Dirichlet-Neumann operators and then by the expansion of the operators we recover the dispersion relation of the system and in addition we obtain the equations of the widely used nonlinear Bousinessq regime. {Although the present study focuses on the three-layer model pertinent to oceanographic dynamics, we emphasize that the analytical framework and methodologies developed herein are broadly applicable to a wide range of scenarios involving interfacial internal waves across stratified fluid layers, extending well beyond the confines of saline ocean environments and also to higher order nonlinear approximations, if the quadratic ones are insufficient for the analysis.}

\section{Model setup and governing equations}\label{model_eq}
We assume that the fluid domain consists of a union of layers $\Omega_i,(i=1,2,3)$ depicted in Figure \ref{fig:system} and defined as follows:
\begin{figure}[h!]
\begin{center}
\fbox{\includegraphics[totalheight=0.35\textheight]{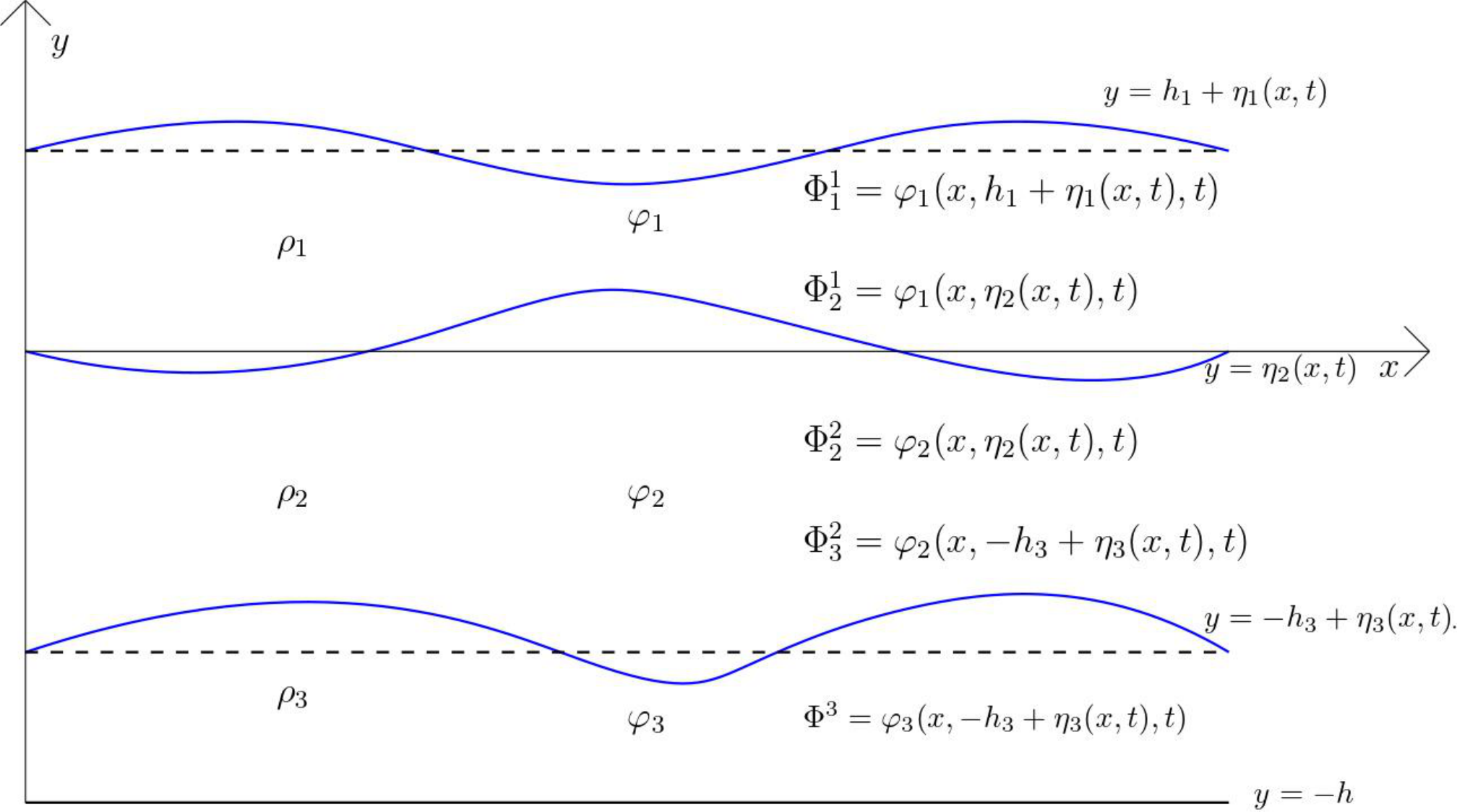}}
\caption{The system under study.}
\label{fig:system}
\end{center}
\end{figure}
\begin{equation}
\setlength{\jot}{10pt}
\begin{split}
&\Omega_1:=\{(x,y,t): x\in\mathbb{R},t\in\mathbb{R},\eta_2(x,t)<y<h_1+\eta_1(x,t)\},\\
&\Omega_2:=\{(x,y,t): x\in\mathbb{R},t\in\mathbb{R},-h_3+\eta_3(x,t)<y<\eta_2(x,t)\},\\
&\Omega_3:=\{(x,y,t): x\in\mathbb{R},t\in\mathbb{R},-h<y<-h_3+\eta_3(x,t)\},
\end{split}
\end{equation}
where $h_1,h_3,h$ are positive constants such that $h>h_3$ and
$\eta_i(x,t)\in\mathcal{S}(\mathbb{R})$, for $i=1,2,3$, are Schwartz functions such that for all $i=1,2,3$ the relations
\begin{equation}\label{zero_av}
    \int_{\mathbb{R}}\eta_i(x,t)\,dx=0\quad{\rm for}\,\,{\rm all}\,\,t,
\end{equation}
hold. Relations \eqref{zero_av} ensure that the average surface level of $\eta_1$ is at $y=h_1$, the average surface level of $\eta_2$ is at $y=0$, and the average surface level of $\eta_3$ is at $y=-h_3$, respectively.

The density of the fluid is assumed to be piecewise constant being equal to $\rho_i (i=1,2,3)$ in the layer $\Omega_i$ for $i=1,2,3$, and such that $\rho_1<\rho_2<\rho_3$ to account for the stable stratification. The same indicial notation will be used to denote the other physical
quantities (e.g. the velocity field, the velocity potentials, the pressure). However, when we refer to the overall physical variable without specification of the layer, we shall use the hat or caret-shaped symbol placed on top of the variable. For example, $(\hat{u},\hat{v})$ will denote the velocity field which takes the value $(u_i,v_i)$ in the layer $\Omega_i$, for $i=1,2,3$. Denoting with $\hat{P}$ the pressure, the equations of motion for an incompressible, inviscid flow are Euler's equations
\begin{equation}\label{Euler}
\begin{split}
    \hat{u}_t+\hat{u}\hat{u}_x+\hat{v}\hat{u}_y &=-\frac{1}{\hat{\rho}}\hat{P}_x\\
    \hat{v}_t+\hat{u}\hat{v}_x+\hat{v}\hat{v}_y &=-\frac{1}{\hat{\rho}}\hat{P}_y-g
    \end{split}\,\,\textrm{in}\,\, \overline{\Omega_1\cup\Omega_2\cup\Omega_3},
\end{equation}
and the equation of mass conservation
\begin{equation}\label{mass_cons}
    \hat{u}_x+\hat{u}_y=0\,\,\textrm{in}\,\, \overline{\Omega_1\cup\Omega_2\cup\Omega_3}.
\end{equation}
Equations \eqref{Euler} and \eqref{mass_cons} are complemented by appropriate boundary conditions. At the free surface, the dynamic boundary condition
\begin{equation}\label{contP_1}
    P_1=P_{atm}\,\,\textrm{on}\,\,y=h_1+\eta_1(x,t),
\end{equation}
decouples the motion of the water from that of the air above, and, in the same time ensures the balance of forces at the surface. Similarly, the balance of forces at the internal boundaries $y=\eta_2(x,t)$ and $y=-h_3+\eta_3(x,t)$, respectively, is expressed by the conditions
\begin{equation}\label{contP_2}
    P_1=P_2\,\,\textrm{on}\,\,y=\eta_2(x,t),
\end{equation}
and 
\begin{equation}\label{contP_3}
    P_2=P_3\,\,\textrm{on}\,\,y=-h_3+\eta_3(x,t).
\end{equation}
The kinematic boundary conditions, which refer to the flat bed $y=-h$, the two interfaces $y=-h_3+\eta_3(x,t)$ and $y=\eta_2(x,t)$, respectively, and to the free surface 
$y=h_1+\eta_1(x,t)$, ensure that these surfaces are impermeable, thus no fluid particle cross them; they are stated as
\begin{equation}\label{kin_cond}
\setlength{\jot}{10pt}
\begin{split}
& v_1=\eta_{1,t}+u_1\eta_{1,x}\,\,\textrm{on}\,\,y=h_1+\eta_1(x,t)\\
  & v_1=\eta_{2,t}+u_1\eta_{2,x}\,\,\textrm{and}\,\,v_2=\eta_{2,t}+u_2\eta_{2,x}   \,\,\textrm{on}\,\,y=\eta_2(x,t)\\
  &v_2=\eta_{3,t}+u_2\eta_{3,x}\,\,\textrm{and}\,\,v_3=\eta_{3,t}+u_3\eta_{3,x}   \,\,\textrm{on}\,\,y={-h_3}+\eta_3(x,t)\\
  &v_3=0\,\,\textrm{on}\,\,y=-h.
\end{split}
\end{equation}
Throughout the paper we will assume that the fluid is rotation free, that is  
\begin{equation}\label{irrot}
    \hat{u}_y(x,y,t)-\hat{v}_x(x,y,t)=0\,\,{\rm for}\,\,{\rm all}\,\,(x,y,t)\in\overline{\Omega_1\cup\Omega_2\cup\Omega_3}.
\end{equation}
The irrotationality assumption \eqref{irrot} allows us to introduce now in each layer $\Omega_i$  a velocity potential $\varphi_i\in C^1(\Omega_i)$ which satisfies 
\begin{equation}\label{vel_pot}
    u_i=\varphi_{i,x}\quad \textrm{and}\quad v_i=\varphi_{i,y}\quad\textrm{for}\,\,i=1,2,3.
\end{equation}
\begin{rem}\label{harm_vel}
Note that, due to \eqref{mass_cons}, the velocity potentials $\varphi_i,\,(i=1,2,3)$ are harmonic functions in $\Omega_i,\,i=1,2,3.$
\end{rem}
We will use, in due course of the paper, the subscripts $_{s_1},_{s_2}$ and $ _{s_3}$, respectively, in order to denote evaluations of various quantities on the surfaces
$y=h_1+\eta_1(x,t)$, $y=\eta_2(x,t)$ and $y=-h_3+\eta_3(x,t)$, respectively.
The kinematic boundary conditions can now be written as 
\begin{align}
    &\eta_{1,t}=(\varphi_{1,y})_{s_1}-\eta_{1,x}(\varphi_{1,x})_{s_1}\\
    &\eta_{2,t}=(\varphi_{1,y})_{s_2}-\eta_{2,x}(\varphi_{1,x})_{s_2}
    =(\varphi_{2,y})_{s_2}-\eta_{2,x}(\varphi_{2,x})_{s_2} \label{eta2t}\\
    &\eta_{3,t}=(\varphi_{2,y})_{s_3}-\eta_{3,x}(\varphi_{2,x})_{s_3}=
    (\varphi_{3,y})_{s_3}-\eta_{3,x}(\varphi_{3,x})_{s_3}.
\end{align}

Moreover, the Euler equations can be recast as
\begin{equation}
\setlength{\jot}{10pt}
\left\{  \begin{split}
    &\nabla\left(\varphi_{1,t}+\frac{|\nabla\varphi_1|^2}{2}+\frac{P_1}{\rho_1}+gy\right)=0\,\,
    \textrm{in}\,\,\Omega_1\\
     &\nabla\left(\varphi_{2,t}+\frac{|\nabla\varphi_2|^2}{2}+\frac{P_2}{\rho_2}+gy\right)=0\,\,
    \textrm{in}\,\,\Omega_2\\
     &\nabla\left(\varphi_{3,t}+\frac{|\nabla\varphi_3|^2}{2}+\frac{P_3}{\rho_3}+gy\right)=0\,\,
    \textrm{in}\,\,\Omega_3\\
  \end{split}  \right.
\end{equation}
where $\nabla$ denotes the spatial gradient. Thus, there are functions $t\mapsto f_i(t), i=1,2,3$ such that 

\begin{equation}\label{Bernoulli}
\setlength{\jot}{10pt}
\left\{ \begin{split}
    &\varphi_{1,t}+\frac{|\nabla\varphi_1|^2}{2}+\frac{P_1}{\rho_1}+gy=f_1(t)\,\,
    \textrm{in}\,\,\Omega_1,\\
     &\varphi_{2,t}+\frac{|\nabla\varphi_2|^2}{2}+\frac{P_2}{\rho_2}+gy=f_2(t)\,\,
    \textrm{in}\,\,\Omega_2,\\
     &\varphi_{3,t}+\frac{|\nabla\varphi_3|^2}{2}+\frac{P_3}{\rho_3}+gy=f_3(t)\,\,
    \textrm{in}\,\,\Omega_3.\\
  \end{split}  \right.
\end{equation}
\\
Utilizing the interface conditions \eqref{contP_1}, \eqref{contP_2} and \eqref{contP_3} and choosing convenient time dependent functions $f_1,\,f_2$ and $f_3$ in \eqref{Bernoulli} we obtain the relations
\\
\begin{equation}\label{Ber}
\setlength{\jot}{10pt}
\left\{ \begin{split}
&\rho_1(\varphi_{1,t})_{s_1}+\rho_1 \frac{|\nabla\varphi_1|^2_{s_1}}{2}+g\rho_1\eta_1=0,\\
&\rho_1\left((\varphi_{1,t})_{s_2}+\frac{|\nabla\varphi_1|^2_{s_2}}{2}+g\eta_2 \right)=
\rho_2\left((\varphi_{2,t})_{s_2}+\frac{|\nabla\varphi_2|^2_{s_2}}{2}+g\eta_2 \right),\\
&\rho_2\left((\varphi_{2,t})_{s_3}+\frac{|\nabla\varphi_2|^2_{s_3}}{2}+g\eta_3 \right)=
\rho_3\left((\varphi_{3,t})_{s_3}+\frac{|\nabla\varphi_3|^2_{s_3}}{2}+g\eta_3 \right).
  \end{split}  \right.
\end{equation}
One of the main tenets of the Hamiltonian formulation refers to the possibility of reducing the number of the dynamical variables. This feature emerges, using the Ansatz by Benjamin \& Bridges \cite{BenBri, BenBri2} (utilized also by Craig et al. \cite{CGS}). Let us define
\begin{equation}\label{traces}
\setlength{\jot}{10pt}
\begin{split}
&\Phi^1_1:=\vp_1(x,h_1+\eta_1(x,t),t)=(\vp_1)_{s_1},\quad\Phi^1_2:=\vp_1(x,\eta_2(x,t),t)=(\vp_1)_{s_2},\\
&\Phi^2_2:=\vp_2(x,\eta_2(x,t),t)=(\vp_2)_{s_2},\quad\Phi^2_3:=\vp_2(x,-h_3+\eta_3(x,t),t)=(\vp_2)_{s_3},\\
&\Phi^3:=\vp_3(x,-h_3+\eta_3(x,t),t)=(\vp_3)_{s_3},\\
&\xi_1:=\rho_1\Phi^1_1,\quad \xi_2:=\rho_2\Phi^2_2-\rho_1\Phi^1_2,\quad \xi_3:=\rho_3\Phi^3-\rho_2\Phi^2_3.
\end{split}
\end{equation}
In what follows it will be shown that the five traces $\Phi^i_j$ can be expressed by the three Hamiltonian ``momenta" $\xi_i$ only. The $\eta_i$ variables will play the role of Hamiltonian ``coordinates". This approach not only reduces the number of variables, but also the dimensionality of the problem, because instead of working in the bulk of the fluid (in the two dimensional spatial domain) one can work only with variables, defined on the surfaces and the interfaces, that is, in one spatial dimension.
 We note that definition \eqref{traces} arises naturally from the Legendre transform based on the Lagrangian, cf. formula (2.15) in Craig et al. \cite{CGS}.
\\
Utilizing the notation in \eqref{traces} we obtain from \eqref{Ber} that $\xi_i$ satisfy the evolution equations
\begin{equation}\label{E1}
\xi_{1,t} -\rho_1\eta_{1,t}(\vp_{1,y})_{s_1}+\frac{\rho_1}{2}|\nabla \vp_1|^2_{s_1}+g\rho_1\eta_1 =0,
\end{equation}
\begin{equation}\label{E2}
\setlength{\jot}{10pt}
\begin{split}
     \xi_{2,t}+\eta_{2,t}[\rho_1(\varphi_{1,y})_{s_2}-\rho_2(\varphi_{2,y})_{s_2}]+
     \rho_2\frac{|\nabla\varphi_2|^2_{s_2}}{2}-\rho_1\frac{|\nabla\varphi_1|^2_{s_2}}{2}\\
     +(\rho_2-\rho_1)g\eta_2=0,
   \end{split}
\end{equation}

\begin{equation}\label{E3}
\setlength{\jot}{10pt}
    \begin{split}
        \xi_{3,t}+\eta_{3,t}[\rho_2(\varphi_{2,y})_{s_3}-\rho_3(\varphi_{3,y})_{s_3}]+
\rho_3\frac{|\nabla\varphi_3|^2_{s_3}}{2}-\rho_2\frac{|\nabla\varphi_2|^2_{s_3}}{2}\\
+(\rho_3-\rho_2)g\eta_3=0.
    \end{split}
\end{equation}
\color{black}
\section{The Hamiltonian formulation of the governing equations and the Dirichlet-Neumann operators}\label{reform_DN}
\noindent In order to alleviate the difficulties brought about by the rich structure of the nonlinear water wave problem a reformulation of the latter is needed. This is the Hamiltonian formulation which provides a choice of new dependent and independent variables in which the equations take their simplest form.
\subsection{Hamiltonian form of the equations. Equations for the moving surfaces}
\noindent This section is devoted to a reformulation of the governing equations and of the boundary conditions \eqref{Euler}-\eqref{kin_cond} by collapsing the dynamical variables into a  one-dimensional representation. More precisely, we will rewrite \eqref{Euler}-\eqref{kin_cond} as a Hamiltonian system of partial differential equations of the form (cf. \cite{Olv})
$$\partial_t\mathfrak{m}=J\frac{\delta H}{\delta\mathfrak{m}},$$
where $t\mapsto\mathfrak{m}(t)$ denotes a path in a Hilbert space $\mathcal{H}$ equipped with an inner product $(\cdot,\cdot)$, the associated \emph{Hamiltonian functional} $H:\mathcal{D}\subset\mathcal{H}\rightarrow \mathbb{R}$ is defined on a dense subspace $\mathcal{D}$ of $\mathcal{H}$. In our considerations 
$\mathfrak{m}=(\eta_1,\eta_2,\eta_3,\xi_1,\xi_2,\xi_3)^T$, the Hilbert space is an appropriate Cartesian product of spaces of square-integrable functions and 
$\mathcal{D}=\mathcal{S}(\mathbb{R})\times\mathcal{S}(\mathbb{R})\times\mathcal{S}(\mathbb{R})\times\mathcal{S}(\mathbb{R})\times\mathcal{S}(\mathbb{R})\times\mathcal{S}(\mathbb{R})$, where $\mathcal{S}(\mathbb{R})$ is the subspace of Schwartz functions in the variable $x$.
Moreover, $J=\begin{pmatrix}0 & I_3\\-I_3 &0 \end{pmatrix}$ ($I_3$ being the $3\times3$ identity matrix), and defines a \emph{Poisson bracket} for functionals by the formula
$$\{F_1,F_2\}=\left(\frac{\delta F_1}{\delta\mathfrak{m}},J\frac{\delta F_2}{\delta\mathfrak{m}}\right).$$
We also denote by
$\frac{\delta H}{\delta\mathfrak{m}}$ the variational derivative of $H$ at $\mathfrak{m}\in\mathcal{D}$ through the formula
$$\lim_{\varepsilon\to 0}\frac{H(\mathfrak{m}+\varepsilon\mathfrak{m}_0)-H(\mathfrak{m})}{\varepsilon}=\left(\mathfrak{m}_0,\frac{\delta H}{\delta \mathfrak{m}}\right)\quad {\rm for}\quad \mathfrak{m}_0\in\mathcal{D},$$
where $(\cdot,\cdot)$ is the inner product in the Hilbert space $L^2(\mathbb{R})\times L^2(\mathbb{R})\times L^2(\mathbb{R})\times L^2(\mathbb{R})\times L^2(\mathbb{R})\times L^2(\mathbb{R})$.

\subsection{The Hamiltonian functional}
The natural candidate for the Hamiltonian functional is the total energy of the flow given as
$E_K+E_P$, where $E_K$ stands for the kinetic energy and $E_P$ denotes the potential energy. In terms of the quantities of the flow, the kinetic energy is written as
\begin{equation}
E_K=\int\int_{\Omega_1\cup\Omega_2\cup\Omega_3}\hat{\rho}\left(\frac{\hat{u}^2+\hat{v}^2}{2}\right)\, dydx
\end{equation}
which, due to stratification and \eqref{vel_pot}, can be detailed as
\begin{equation}\label{kin_energy}
\setlength{\jot}{10pt}
\begin{split}
E_K=&\int_{\mathbb{R}}\int_{-h}^{-h_3+\eta_3(x,t)}\rho_3\left(\frac{u_3^2+v_3^2}{2}\right)\,dy dx
+\int_{\mathbb{R}}\int_{-h_3+\eta_3(x,t)}^{\eta_2(x,t)}\rho_2\left(\frac{u_2^2+v_2^2}{2}\right)\,dy dx\\
&+\int_{\mathbb{R}}\int_{\eta_2(x,t)}^{h_1+\eta_1(x,t)}\rho_1\left(\frac{u_1^2+v_1^2}{2}\right)\,dy dx\\
=&\frac{\rho_3}{2}\int_{\mathbb{R}}\int_{-h}^{-h_3+\eta_3(x,t)}|\nabla\varphi_3|^2\,dy dx+
\frac{\rho_2}{2}\int_{\mathbb{R}}\int_{-h_3+\eta_3(x,t)}^{\eta_2(x,t)}|\nabla\varphi_2|^2\,dy dx\\
&+\frac{\rho_1}{2}\int_{\mathbb{R}}\int_{\eta_2(x,t)}^{h_1+\eta_1(x,t)}|\nabla\varphi_1|^2\,dy dx.\\
\end{split}
\end{equation}
In defining the potential energy we act on the realization that
 not all the latter energy can be converted into kinetic energy. However, the difference in the potential energy between the perturbed wave state and the pure background state is available for this conversion. Hence, we have
\begin{equation}\label{pot_energy}
\begin{split}
 E_P=&g\rho_3\int_{\mathbb{R}}\left(\int_{-h}^{-h_3+\eta_3}y\,dy-\int_{-h}^{-h_3}y\,dy\right)+g\rho_2\int_{\mathbb{R}}\left(\int_{-h_3+\eta_3}^{\eta_2}y\,dy-\int_{-h_3}^0y\,dy\right)\\[1em]
 &+g\rho_1\int_{\mathbb{R}}\left(\int_{\eta_2}^{h_1+\eta_1}y\,dy-\int_0^{h_1}y\,dy\right)\\[1em]
 =&\frac{g\rho_1}{2}\int_{\mathbb{R}}\eta_1^2\,dx+\frac{g(\rho_2-\rho_1)}{2}\int_{\mathbb{R}}\eta_2^2\,dx+\frac{g(\rho_3-\rho_2)}{2}\int_{\mathbb{R}}\eta_3^2\,dx.
  \end{split}  
\end{equation}
We would like to note that the above choice for the potential energy allows us to avoid any non-physical constant terms in the energy density and hence infinite energy. More specifically, the potential energy functional \eqref{pot_energy} is not convergent if the potential energy of the pure current background state is not subtracted.

\subsection{Computations of variations of $H$} We first list a few useful formulas for computing the variations of $H$.
\begin{rem}
\begin{enumerate}
\item If $F(x)=\int_{g_1(x)}^{g_2(x)}f(y)\,dy$ for given real functions of one variable $x\mapsto f(x)$, $x\mapsto g_1(x)$ and $x\mapsto g_2(x)$, then 
\begin{equation}\label{var_int}
    (\delta F)(x)=\int_{g_1(x)}^{g_2(x)}(\delta f)(y)\,dy+f(g_2(x))(\delta g_2)(x)-f(g_1(x))(\delta g_1)(x).
\end{equation}
\item For given functions $(x,y)\mapsto F(x,y)$, $x\mapsto f_i(x),\,(i=1,2)$ it holds
\begin{equation}\label{deriv_var_lim_int}
\setlength{\jot}{10pt}
    \begin{split}
        \frac{\partial}{\partial x}\left(\int_{f_1(x)}^{f_2(x)}F(x,y)\,dy\right)=&\int_{f_1(x)}^{f_2(x)}F_x(x,y)\,dy\\
       & +F(x,f_2(x))f_2'(x)-F(x,f_1(x))f_1'(x)
    \end{split}
\end{equation}
    \item If $\theta_1$ and $\theta_2$ are harmonic functions satisfying $\nabla\cdot(\theta_1\nabla\theta_2)=(\nabla\theta_1)\cdot(\nabla\theta_2)$, then for harmonic variations $\delta\varphi_i$ of $\varphi_i$ it holds that
    \begin{equation}\label{var_grad_sq}
    \delta\big((\nabla\varphi_i)\cdot (\nabla\varphi_i)\big)=2\nabla\cdot(\delta\varphi_i\nabla\varphi_i),
    \end{equation}
    since
    $|\nabla \varphi_i|^2=\nabla\cdot (\varphi_i \nabla \varphi_i)$ for $i=1,2,3$.
\end{enumerate}
\end{rem}
\begin{rem}
 Regarding variations of the traces on the interfaces of the velocity potentials the following formulas are true.
\begin{enumerate}
    \item On the free surface $y=h_1+\eta_1(x,t)$ we have
    \begin{equation}\label{var_pot_1}
    \setlength{\jot}{10pt}
        \begin{split}
        (\delta &\Phi_1^1)(x,t)=\\
        =&\lim_{\varepsilon\to 0}\frac{(\varphi_1+\varepsilon\delta\varphi_1)\big(x,h_1+\eta_1(x,t)+\varepsilon(\delta\eta_1)(x,t),t\big)-\varphi_1(x,h_1+\eta(x,t),t)}{\varepsilon}\\
        =&(\varphi_{1,y})_{s_1}\delta\eta_1+(\delta\varphi_1)_{s_1}
        \end{split}
        \end{equation}
    \item Analogously on the interface $y=\eta_2(x,t)$ it holds
        \begin{equation}\label{var_pot_2}
            \delta\Phi^1_2=(\varphi_{1,y})_{s_2}\delta\eta_2+(\delta\varphi_1)_{s_2}\quad{\rm and}\quad \delta\Phi_2^2=(\varphi_{2,y})_{s_2}\delta\eta_2+(\delta\varphi_2)_{s_2}
        \end{equation}
        \item Lastly, on the interface $y=-h_3+\eta_3(x,t)$ we have
        \begin{equation}\label{var_pot_3}
            \delta\Phi^2_3=(\varphi_{2,y})_{s_3}\delta\eta_3+(\delta\varphi_2)_{s_3}\quad {\rm and}\quad (\varphi_{3,y})_{s_3}\delta\eta_3+(\delta\varphi_3)_{s_3}
        \end{equation}
\end{enumerate}
\end{rem}
We proceed now to compute the variations of the terms that make up the kinetic energy $E_K$ from
\eqref{kin_energy}. Utilizing the divergence theorem (\cite{Eva}), formula \eqref{var_pot_3} as well as \eqref{var_int}, \eqref{deriv_var_lim_int} and \eqref{var_grad_sq} we obtain
\begin{equation}\label{var_grad1}
\setlength{\jot}{10pt}
\begin{split}
\delta\left(\int_{\mathbb{R}}\int_{-h}^{-h_3+\eta_3}|\nabla \varphi_3|^2\,dydx\right)=&2 \int_{\mathbb{R}}(\vp_{3,y}-\eta_{3,x}\vp_{3,x})_{s_3} \left(\delta\Phi^3-(\delta\eta_3)(\vp_{3,y})_{s_3}\right)dx\\
&+\int_{\mathbb{R}}|\nabla\vp_3|^2_{s_3}\delta\eta_3\, dx.
\end{split}
\end{equation}
Analogously, by the same token as above we have
\begin{equation}\label{var_grad2}
\setlength{\jot}{10pt}
\begin{split}
\delta\left(\int_{\mathbb{R}}\int_{-h_3+\eta_3}^{\eta_2}|\nabla \varphi_2|^2\,dydx\right)=&2 \int_{\mathbb{R}} (\vp_{2,y}-\eta_{2,x}\vp_{2,x})_{s_2} \left(\delta\Phi^2_2-(\delta\eta_2)(\vp_{2,y})_{s_2}\right) dx\\
&-2 \int_{\mathbb{R}} (\vp_{2,y}-\eta_{3,x}\vp_{2,x})_{s_3} \left(\delta\Phi^2_3-(\delta\eta_3)(\vp_{2,y})_{s_3}\right) dx\\
&+\int_{\mathbb{R}}|\nabla\vp_2|^2_{s_2}\delta\eta_2\,dx -\int_{\mathbb{R}}|\nabla\vp_2|^2_{s_3}\delta\eta_3\,dx,
\end{split}
\end{equation}
and 
\begin{equation}\label{var_grad3}
\setlength{\jot}{10pt}
\begin{split}
\delta\left(\int_{\mathbb{R}}\int_{\eta_2}^{h_1+\eta_1}|\nabla \varphi_1|^2\,dydx\right)=&2 \int_{\mathbb{R}} (\vp_{1,y}-\eta_{1,x}\vp_{1,x})_{s_1} \left(\delta\Phi^1_1-(\delta\eta_1)(\vp_{1,y})_{s_1}\right) dx\\
&-2 \int_{\mathbb{R}} (\vp_{1,y}-\eta_{2,x}\vp_{1,x})_{s_2} \left(\delta\Phi^1_2-(\delta\eta_2)(\vp_{1,y})_{s_2}\right) dx\\
&+\int_{\mathbb{R}}|\nabla\vp_1|^2_{s_1}\delta\eta_1\,dx -\int_{\mathbb{R}}|\nabla\vp_1|^2_{s_2}\delta\eta_2\,dx.
\end{split}
\end{equation}

Collecting together the variations of $H$ (from \eqref{var_grad1}-\eqref{var_grad3} and \eqref{pot_energy}) with respect to $\eta_1$, $\eta_2$ and $\eta_3$ respectively, we find that 
\begin{equation}\label{var_eta1}
\setlength{\jot}{10pt}
\begin{split}
\frac{\delta H}{\delta\eta_1}=&-\rho_1\eta_{1,t}(\vp_{1,y})_{s_1}+\frac{\rho_1}{2}|\nabla \vp_1|^2_{s_1}+g\rho_1\eta_1\\
=&-\rho_1\eta_{1,t}(\vp_{1,y})_{s_1}-\rho_1(\vp_{1,t})_{s_1}\\
=&-\frac{d}{dt}\left(\rho_1\vp_1(x,h_1+\eta_1(x,t),t)\right)\\
=&-\xi_{1,t},
\end{split}
\end{equation}

\begin{equation}\label{var_eta2}
\setlength{\jot}{10pt}
\begin{split}
\frac{\delta H}{\delta\eta_2}=&\rho_1\eta_{2,t}(\vp_{1,y})_{s_2}-\frac{\rho_1}{2}|\nabla\vp_1|^2_{s_2}-\rho_2\eta_{2,t}(\vp_{2,y})_{s_2}+\frac{\rho_2}{2}|\nabla\vp_2|^2_{s_2}+(\rho_2-\rho_1)g\eta_2\\
=&-\rho_2\left((\vp_{2,t})_{s_2}+\eta_{2,t}(\vp_{2,y})_{s_2}\right)+\rho_1\left((\vp_{1,t})_{s_2}+\eta_{2,t}(\vp_{1,y})_{s_2}\right)\\
=&-\frac{d}{dt}\left(\rho_2\vp_2(x,\eta_2(x,t),t)\right)+\frac{d}{dt}\left(\rho_1\vp_1(x,\eta_2(x,t),t)\right)\\
=&-\xi_{2,t},
\end{split}
\end{equation}
and
\begin{equation}\label{var_eta3}
\setlength{\jot}{10pt}
\begin{split}
\frac{\delta H}{\delta\eta_3}=&\rho_2\eta_{3,t}(\vp_{2,y})_{s_3}-\frac{\rho_2}{2}|\nabla\vp_2|^2_{s_3}-\rho_3\eta_{3,t}(\vp_{3,y})_{s_3}+\frac{\rho_3}{2}|\nabla\vp_3|^2_{s_3}+(\rho_3-\rho_2)g\eta_3\\
=&-\rho_3\left((\vp_{3,t})_{s_3}+\eta_{3,t}(\vp_{3,y})_{s_3}\right)+\rho_2\left((\vp_{2,t})_{s_3}+\eta_{3,t}(\vp_{2,y})_{s_3}\right)\\
=&-\frac{d}{dt}\left(\rho_3\vp_3(x,-h_3+\eta_3(x,t),t)\right)+\frac{d}{dt}\left(\rho_2\vp_2(x,-h_3+\eta_3(x,t),t)\right)\\
=&-\xi_{3,t}.
\end{split}
\end{equation}
Moreover, glancing at \eqref{traces} and utilizing \eqref{var_grad1}-\eqref{var_grad3} we infer that
\begin{equation}\label{var_xis}
\frac{\delta H}{\delta \xi_1}=\eta_{1,t},\quad \frac{\delta H}{\delta \xi_2}=\eta_{2,t},\quad \frac{\delta H}{\delta \xi_3}=\eta_{3,t}.
\end{equation}
We summarize the conclusions regarding the variations of $H$ computed in \eqref{var_eta1}, \eqref{var_eta2}, \eqref{var_eta3} and \eqref{var_xis} in the following theorem.
\begin{thm}
The governing equations and boundary conditions \eqref{Euler}-\eqref{kin_cond} posses the Hamiltonian formulation
\begin{equation}\label{Ham_syst}
\left\{
\begin{split}
\frac{\delta H}{\delta\eta_1}=-\xi_{1,t},\quad \frac{\delta H}{\delta \xi_1}=\eta_{1,t},\\
\frac{\delta H}{\delta\eta_2}=-\xi_{2,t},\quad \frac{\delta H}{\delta \xi_2}=\eta_{2,t},\\
\frac{\delta H}{\delta\eta_3}=-\xi_{3,t},\quad \frac{\delta H}{\delta \xi_3}=\eta_{3,t}.\\
\end{split}
\right.
\end{equation}
\end{thm}
\begin{rem}
 The (equivalent) Hamiltonian (re)formulation of the water wave problem not only reduces the number of the dynamical variables, but (together with the Dirichlet-Neumann machinery) also offers the possibility to derive (simpler) nonlinear model equations directly from the full nonlinear water waves equations in multi-layer flows with a free surface: as a result, the coefficients in these equations (and also their solutions) depend explicitly on physical parameters on the multi-layer system, cf. \cite{CGK, CGS, CICMP}. More precisely, from the a single overarching principle--the Hamiltonian formulation--one can derive model equations pertaining to various physical regimes by simply defining suitable scaling regimes in terms of certain ratios like amplitude/layer depth and/or amplitude/wavelength, cf. e.g. \cite{CGK, CGS, CICMP}. This is in stark contrast with the usual case-by-case analysis \cite{Whi}.
Furthermore, the Hamiltonian perspective is very suited to perturbation analysis since it guarantees that if appropriate symmetries are preserved in the approximation process, then conservation properties of the exact system are retained in its approximations. 
\end{rem}
\subsection{The Dirichlet-Neumann operators}

The Dirichlet-Neumann (DN) operator corresponding to the uppermost layer is defined as
\begin{equation}\label{dn_u}
\begin{pmatrix}G_{11} & G_{12}\\ G_{21} & G_{22}\end{pmatrix}\begin{pmatrix}(\varphi_{1})_{s_2}\\ (\varphi_{1})_{s_1}\end{pmatrix}=
\begin{pmatrix} -\sqrt{1+\eta_{2,x}^2}\left(\frac{\partial \varphi_1}{\partial \mathbf{n}_2} \right)_{s_2} \\
\sqrt{1+\eta_{1,x}^2}\left(\frac{\partial \varphi_1}{\partial \mathbf{n}_1} \right)_{s_1}
\end{pmatrix}
\end{equation}
where $\mathbf{n}_2=\frac{(-\eta_{2,x},1)}{\sqrt{1+\eta_{2,x}^2}}$ and $\mathbf{n}_1=\frac{(-\eta_{1,x},1)}{\sqrt{1+\eta_{1,x}^2}}$ are the unit normal vectors to the upper interface $y=\eta_2(x,t)$  and to the top surface $y=-h_1+\eta_1(x,t)$, respectively. 
The DN operator corresponding to the middle layer is
\begin{equation}\label{dn_m}
\begin{pmatrix}\Gamma_{11} & \Gamma_{12}\\ \Gamma_{21} &\Gamma_{22}\end{pmatrix}\begin{pmatrix}(\varphi_2)_{s_2}\\ (\varphi_2)_{s_3}\end{pmatrix}=
\begin{pmatrix} \sqrt{1+\eta_{2,x}^2}\left(\frac{\partial \varphi_2}{\partial \mathbf{n}_2} \right)_{s_2} \\
-\sqrt{1+\eta_{3,x}^2}\left(\frac{\partial \varphi_2}{\partial \mathbf{n}_3} \right)_{s_3}
\end{pmatrix}
\end{equation}
where $\mathbf{n}_3=\frac{(-\eta_{3,x},1)}{\sqrt{1+\eta_{3,x}^2}}$ denotes the outward pointing unit normal vector corresponding to the lower interface $y=-h_3+\eta_3(x,t)$. 
Moreover, the DN operator corresponding to the bottom layer is
\begin{equation}\label{dn_b}
 G(\eta_3)(\varphi_3)_{s_3}=\sqrt{1+\eta_{3,x}^2}\left(\frac{\partial\varphi_3}{\partial \mathbf{n}_3}\right)_{s_3}=(\nabla\varphi_3)_{s_3}\cdot(-\eta_{3,x},1).
\end{equation}

These DNOs for each domain $\Omega_1,$ $\Omega_2,$  $\Omega_3$ could be obtained recursively in the form of asymptotic series in $\eta_1, \, \eta_2$ and $\eta_3$ following the technique described in the Appendix of \cite{CGK}, see also Appendix A of \cite{IMT}.

Appealing to the divergence theorem (\cite{Eva}) the kinetic energy of the system \eqref{kin_energy} can be rewritten as
\begin{equation}
\begin{split}
 \rm{E_K}=&\frac{\rho_1}{2}\int(\varphi_1)_{s_1}\left(\nabla\varphi_1\right)_{s_1}\cdot (-\eta_{1,x},1)\,dx
 -\frac{\rho_1}{2}\int(\varphi_1)_{s_2}\left(\nabla\varphi_1\right)_{s_2}\cdot (-\eta_{2,x},1)\,dx\\
&+ \frac{\rho_2}{2}\int(\varphi_2)_{s_2}\left(\nabla\varphi_2\right)_{s_2}\cdot (-\eta_{2,x},1)\,dx
-\frac{\rho_2}{2}\int(\varphi_2)_{s_3}\left(\nabla\varphi_2\right)_{s_3}\cdot (-\eta_{3,x},1)\,dx\\
&+\frac{\rho_3}{2}\int(\varphi_3)_{s_3}\left(\nabla\varphi_3\right)_{s_3}\cdot (-\eta_{3,x},1)\,dx.
 \end{split}
\end{equation}
With the help of the DN operators, the kinetic energy can be written as
\begin{equation}\label{Kin}
 \begin{split}
\rm{E_K}=&\frac{\rho_1}{2}\int(\varphi_1)_{s_1}\left[G_{21}(\varphi_1)_{s_2}+G_{22}(\varphi_1)_{s_1}\right]\,dx  \\
&+\frac{\rho_1}{2}\int(\varphi_1)_{s_2}\left[G_{11}(\varphi_1)_{s_2}+G_{12}(\varphi_1)_{s_1}\right]\,dx\\
&+\frac{\rho_2}{2}\int (\varphi_2)_{s_2}\left[\Gamma_{11}(\varphi_2)_{s_2} +\Gamma_{12}(\varphi_2)_{s_3}\right]\,dx\\
&+\frac{\rho_2}{2}\int (\varphi_2)_{s_3}\left[\Gamma_{21}(\varphi_2)_{s_2} +\Gamma_{22}(\varphi_2)_{s_3}\right]\,dx\\
&+\frac{\rho_3}{2}\int(\varphi_3)_{s_3} G(\eta_3)(\varphi_3)_{s_3}\,dx.
 \end{split}
\end{equation}
Attempting to simplify the expression \eqref{Kin} we note first that, from the kinematic condition at the upper interface $y=\eta_2(x,t)$, we have $$\eta_{2,t}=(\nabla\varphi_2)_{s_2}\cdot(\eta_{2,x},-1)=(\nabla\varphi_1)_{s_2}\cdot (\eta_{2,x},-1),$$ which, in terms of the DN operators, becomes
\begin{equation}\label{GammaGij}
 \eta_{2,t}=-(\Gamma_{11}(\varphi_2)_{s_2}+\Gamma_{12}(\varphi_2)_{s_3})=G_{11}(\varphi_1)_{s_2}+G_{12}(\varphi_1)_{s_1}.
\end{equation}
Moreover, 
\begin{equation}
 \eta_{3,t}=(\varphi_{2,y}-\varphi_{2,x}\eta_{3,x})_{s_3}=(\varphi_{3,y}-\varphi_{3,x}\eta_{3,x})_{s_3},
 \end{equation}
 which implies
 \begin{equation}\label{GammaG}
 \eta_{3,t}=G(\eta_3)(\varphi_3)_{s_3}=-[\Gamma_{21}(\varphi_2)_{s_2} +\Gamma_{22}(\varphi_2)_{s_3}].
\end{equation}
Recalling that 
\begin{equation}\label{xiDef}
   \xi_1=\rho_1(\varphi_1)_{s_1},\,\,\xi_2=\rho_2 (\varphi_2)_{s_2}-\rho_1 (\varphi_1)_{s_2},\,\,\xi_3=\rho_3 (\varphi_3)_{s_3}-\rho_2(\varphi_2)_{s_3},
    \end{equation}
    and employing \eqref{GammaGij} and \eqref{GammaG} we see that the kinetic energy can be simplified as 
\begin{equation}\label{EK}
\begin{split}
{\rm E_{K}}=&\frac{1}{2}\int\xi_1\left[G_{21}(\varphi_1)_{s_2}+G_{22}(\varphi_1)_{s_1}\right]\,dx  \\
&+\frac{1}{2}\int\xi_2\left[\Gamma_{11}(\varphi_2)_{s_2} +\Gamma_{12}(\varphi_2)_{s_3}\right]\,dx\\
&+\frac{1}{2}\int\xi_3 G(\eta_3)(\varphi_3)_{s_3}\,dx.
 \end{split}
\end{equation} 
Moreover, from \eqref{GammaGij} and \eqref{GammaG} we also obtain the system
\\
\begin{align}  
\xi_2 &=\rho_2 (\varphi_2)_{s_2}+\rho_1 G_{11}^{-1}[\Gamma_{11}(\varphi_2)_{s_2} +\Gamma_{12}(\varphi_2)_{s_3}+\rho_1^{-1}G_{12}\xi_1]\\
\xi_3&=-\rho_2 (\varphi_2)_{s_3}-\rho_3 G^{-1}(\eta_3)[\Gamma_{21}(\varphi_2)_{s_2} +\Gamma_{22}(\varphi_2)_{s_3}] \label{xi_xi2_syst}
\end{align}
\\
which can be rearranged as
\\
\begin{equation}\label{QEq}
\setlength{\jot}{10pt}
\begin{split}
(\rho_2+\rho_1 G_{11}^{-1}\Gamma_{11})(\varphi_2)_{s_2}+(\rho_1 G_{11}^{-1}\Gamma_{12})(\varphi_2)_{s_3} &=\xi_2-G_{11}^{-1}G_{12}\xi_1\\
(-\rho_3 G^{-1}(\eta_3)\Gamma_{21})(\varphi_2)_{s_2} -(\rho_2 +\rho_3 G^{-1}(\eta_3)\Gamma_{22})(\varphi_2)_{s_3} &=\xi_3,
\end{split}
\end{equation}
where the superscript $``^{-1}"$ in the previous system (and in the subsequent equations) is used to refer to the inverse of the operator in question.\\
Taking into account the fact that \eqref{QEq} is a linear system of equations with operator-valued coefficients, we obtain
\begin{equation}\label{solQEq}
\begin{split}
&(\varphi_2)_{s_2}=\mathcal{D}^{-1}\big( \rho_1 G \xi_3 + (\rho_2 G+\rho_3\Gamma_{22})\Gamma_{12}^{-1}(G_{11}\xi_2-G_{12}\xi_1) \big)\equiv \Phi_2^2(\eta_i, \xi_i) \\
&(\varphi_2)_{s_3}=-(\mathcal{D}^*)^{-1}\big( \rho_3(G_{11}\xi_2-G_{12}\xi_1)  +  (\rho_2 G_{11}+\rho_1\Gamma_{11})\Gamma_{21}^{-1} G \xi_3 \big),
\end{split}
\end{equation}
where $\mathcal{D}:=\big( (\rho_2 G+\rho_3\Gamma_{22})\Gamma_{12}^{-1}(\rho_2 G_{11}+\rho_1\Gamma_{11})-\rho_1 \rho_3 \Gamma_{21} \big)$
and the star $(*)$ notation denotes the operation of conjugation \footnote{The conjugation is with respect to the standard real bilinear form $(f,g)=\int_{\mathbb{R}} f(x) g(x) dx.$  }. We used also the fact that $G_{ij}=G^*_{ji},$ $\Gamma_{ij}=\Gamma^*_{ji},$ see for example\cite{CGK}.
Moreover, from \eqref{xiDef} we also have
\begin{equation}\label{phi13surf}
\begin{split}
&(\varphi_1)_{s_2}=\frac{1}{\rho_1}\mathcal{D}^{-1}\big( \rho_1 (\rho_2 G \xi_3 +\rho_3\Gamma_{21} \xi_2)- (\rho_2 G+\rho_3\Gamma_{22})\Gamma_{12}^{-1}(\rho_2 G_{12}\xi_1+\rho_1 \Gamma_{11}\xi_2) \big) \equiv \Phi^1_2(\eta_i, \xi_i) \\
&(\varphi_3)_{s_3}=-(\mathcal{D}^*)^{-1}\big( \rho_2 G_{11}\xi_2-\rho_2 G_{12}\xi_1+\rho_1 \Gamma_{12}\xi_3  \tcr{-}  (\rho_2 G_{11}+\rho_1\Gamma_{11})\Gamma_{21}^{-1} \Gamma_{22} \xi_3 \big).
\end{split}
\end{equation}
The previous expressions of the surface values of the potentials together with relation \eqref{pot_energy} allow us to express the Hamiltonian as a functional of the quantities $\eta_i,\,\xi_i, (i=1,2,3)$, that is, we have
\begin{equation}\label{HamGen}
\begin{split}
H&={\rm E_{K}(\xi_i)}+V(\eta_i)\\
&=\frac{1}{2}\int \xi_1 A_{11} \xi_1 \,dx + \frac{1}{2}\int \xi_2 A_{22} \xi_2 \,dx+\frac{1}{2}\int \xi_3 A_{33} \xi_3 \,dx \\ 
&+ \int \xi_1 A_{12} \xi_2 \, dx + \int \xi_1 A_{13} \xi_3 \, dx + \int \xi_2 A_{23} \xi_3 \, dx    \\
&+\frac{1}{2} \int \left( g\rho_1 \eta_1^2 +g(\rho_2-\rho_1) \eta_2^2 + g(\rho_3-\rho_2) \eta_3^2\right) \, dx,
\end{split}
\end{equation}
where the last integral is the potential energy $V$ and the operators $A_{ij} (i,j=1,2,3)$, written in terms of the DN operators, are
\begin{equation}\label{Aij}
\begin{split}
&A_{11}=\frac{1}{\rho_1}G_{22}- \frac{\rho_2}{\rho_1} G_{21}\mathcal{D}^{-1}(\rho_2 G  +\rho_3 \Gamma_{22})\Gamma_{12}^{-1}G_{12}   ,\\
&A_{22}= \Gamma_{11} \mathcal{D}^{-1} (\rho_2 G + \rho_3 \Gamma_{22})\Gamma_{12}^{-1}G_{11}-\rho_3 G_{11} \mathcal{D}^{-1} \Gamma_{21},     \\
& A_{33}= \Big(  \Gamma_{22} \Gamma_{12}^{-1}(\rho_2 G_{11}+\rho_1 \Gamma_{11}) -\rho_1 \Gamma_{21} \Big)\mathcal{D}^{-1} G,\\
&A_{12}=\rho_3 G_{21} \mathcal{D}^{-1} \Gamma_{21} - \frac{1}{2} G_{21}\Big(  \mathcal{D}^{-1}(\rho_2 G + \rho_3 \Gamma_{22})\Gamma_{12}^{-1}+
\Gamma_{21}^{-1} (\rho_2 G + \rho_3 \Gamma_{22}) (\mathcal{D}^*)^{-1} \Big)\Gamma_{11},\\
&A_{13}=\frac{1}{2}\rho_2 G_{21} \Big(\mathcal{D}^{-1}+(\mathcal{D}^*)^{-1}\Big)G,\\
&A_{23}=\frac{1}{2}\Big( (\rho_1 \Gamma_{11}-\rho_2 G_{11}) \mathcal{D}^{-1} - \Gamma_{12} (\mathcal{D}^*)^{-1}
(\rho_2 G_{11}+\rho_1 \Gamma_{11})\Gamma_{21}^{-1}\Big) G.
\end{split}
\end{equation}
The Hamiltonian equations can be represented, via the DN operators, with the help of the quantities $(\varphi_{1,x}(x,y,t))_{s_1}$ and  $(\varphi_{1,y}(x,y,t))_{s_1}$. From \eqref{HamGen} and because the operators $A_{ij}$ do not depend on $\xi_i$, we obtain (in terms of the Hamiltonian variables) the evolution equation
\begin{equation}\label{eta1}
    \eta_{1,t}=\frac{\delta H}{\delta \xi_1}=A_{11}\xi_1+A_{12}\xi_2+A_{13}\xi_3:=F_1(\eta_i,\xi_i).
\end{equation}
In addition, by \eqref{HamGen} and \eqref{Ham_syst},  $(\varphi_{1,x}(x,y,t))_{s_1}$ and  $(\varphi_{1,y}(x,y,t))_{s_1}$ satisfy the system
\begin{equation}\label{phi_surf}
\begin{split}
 &(\varphi_{1,x})_{s_1}+(\varphi_{1,y})_{s_1} \eta_{1,x}= \frac{1}{\rho_1}\xi_{1,x},\\
  &(\varphi_{1,y})_{s_1}-(\varphi_{1,x})_{s_1}\eta_{1,x}=\eta_{1,t}=F_1(\eta_i,\xi_i).
  \end{split}
\end{equation}
Solving the system we find
\begin{equation}\label{phi_surf sol}
\begin{split}
 & (\varphi_{1,x})_{s_1}=  \frac{\xi_{1,x}-\rho_1\eta_{1,x} F_1}{\rho_1(1+\eta_{1,x}^2)}  ,   \\
  &(\varphi_{1,y})_{s_1}= \frac{\eta_{1,x}\xi_{1,x}+\rho_1 F_1}{\rho_1(1+\eta_{1,x}^2)}.
  \end{split}
\end{equation}
which allow us to write equation \eqref{E1} through the Hamiltonian variables.
More precisely, we have

\begin{equation}\label{E11}
\xi_{1,t} + \frac{\xi_{1,x}^2 -2\rho_1\eta_{1,x}\xi_{1,x}F_1- \rho_1^2F_1^2}{2\rho_1(1+\eta_{1,x}^2)} +g\rho_1\eta_1 =0.
\end{equation}

The counterpart equation for $\eta_1$ is of course \eqref{eta1},

\begin{equation}\label{E12}
\eta_{1,t} = F_1(\xi_i,\eta_i).
\end{equation}



From \eqref{HamGen}, where the operators $A_{ij}$ do not depend on $\xi_i$, we have
\begin{equation}\label{eta2}
    \eta_{2,t}=\frac{\delta H}{\delta \xi_2}= A_{12}^* \xi_1+A_{22}\xi_2+A_{23} \xi_3 =: F_2(\eta_i,\xi_i).
\end{equation} 
Thus $F_2(\eta_i,\xi_i)$ is defined through the canonical variables.
Next, from \eqref{solQEq} we have $(\varphi_2)_{s_2}\equiv \Phi_2^2(\eta_i, \xi_i)$ as a known quantity, defined through the canonical variables.
Since \begin{equation}
    \frac{d}{dx} \Phi_2^2= (\varphi_{2,x})_{s_2}+(\varphi_{2,y})_{s_2}\eta_{2,x}
\end{equation}
and from \eqref{eta2t}
\begin{equation}
    F_2=\eta_{2,t}=(\varphi_{2,y})_{s_2}-(\varphi_{2,x})_{s_2}\eta_{2,x},
\end{equation}
we have a system,
\begin{align}
& (\varphi_{2,x})_{s_2}+ \eta_{2,x}(\varphi_{2,y})_{s_2}=(\Phi^2_2)_x \\
-&\eta_{2,x} (\varphi_{2,x})_{s_2}+(\varphi_{2,y})_{s_2}=F_2
\end{align}
which gives
\begin{align}
    (\varphi_{2,x})_{s_2}&= \frac{(\Phi^2_2)_x-\eta_{2,x}F_2}{1+\eta_{2,x}^2 }
        \\
        (\varphi_{2,y})_{s_2}&= \frac{\eta_{2,x}(\Phi^2_2)_x +F_2}{1+\eta_{2,x}^2 }
\end{align}
and similarly,
\begin{align}
    (\varphi_{1,x})_{s_2}&= \frac{(\Phi^1_2)_x-\eta_{2,x}F_2}{1+\eta_{2,x}^2 }
        \\
        (\varphi_{1,y})_{s_2}&= \frac{\eta_{2,x}(\Phi^1_2)_x +F_2}{1+\eta_{2,x}^2 }
\end{align}
where $\Phi^1_2$ is defined in \eqref{phi13surf}.
Therefore, $\xi_2$ satisfies the equation
\begin{equation} \label{E22}
    \xi_{2,t}+\frac{\rho_2 (\Phi^2_{2,x})^2-\rho_1 (\Phi_{2,x}^1)^2-2F_2\xi_{2,x}\eta_{2,x}+(\rho_1-\rho_2)F_2^2   }{2(1+\eta_{2,x}^2)}+(\rho_2-\rho_1)g \eta_2 =0.
\end{equation}
Similarly, for the last pair of equations we obtain
\begin{equation}\label{eta3}
    \eta_{3,t}=\frac{\delta H}{\delta\xi_3}=A_{13}^{*}\xi_1 +A_{23}^{*}\xi_2+A_{33}\xi_3:=F_3(\eta_i,\xi_i)
\end{equation}
and
\begin{equation}\label{E33}
     \xi_{3,t}+\frac{\rho_3 (\Phi^3_{x})^2-\rho_2 (\Phi_{3,x}^2)^2-2F_3\xi_{3,x}\eta_{3,x}+(\rho_2-\rho_3)F_3^2   }{2(1+\eta_{3,x}^2)}+(\rho_3-\rho_2)g \eta_3 =0.
\end{equation}
We note that equations \eqref{E11}, \eqref{E22}, \eqref{E33}, \eqref{E12}, \eqref{eta2}, \eqref{eta3} render a closed system of six equations for the six variables $(\eta_i,\xi_i), $ $i=1,2,3$. All quantities in these equations are given in terms of the Dirichlet-Neumann operators. Most importantly, equations \eqref{E11}, \eqref{E22}, \eqref{E33}, \eqref{E12}, \eqref{eta2}, \eqref{eta3} also suggest the pattern for more superimposed layers. 

 





\section{The Dirichlet-Neumann operators and the dispersion relation}\label{disp}
\subsection{The derivation of the dispersion relation}
A significant role in the derivation of the dispersion relation will be played by the Dirichlet-Neumann operators introduced in \eqref{dn_u}, \eqref{dn_m} and \eqref{dn_b}. The DN operators are analytic in their dependence on $(\eta_1,\eta_2)$, $(\eta_2,\eta_3)$, and $\eta_3$, respectively, and have convergent Taylor series expansions
\begin{equation}
\begin{pmatrix}
G_{11}(\eta_1,\eta_2) & G_{12}(\eta_1,\eta_2)\\
G_{21}(\eta_1,\eta_2) & G_{22}(\eta_1,\eta_2)
\end{pmatrix}=\sum_{k_1,k_2=0}^{\infty}\begin{pmatrix}
G_{11}^{(k_1 k_2)}(\eta_1,\eta_2) & G_{12}^{(k_1 k_2)}(\eta_1,\eta_2)\\
G_{21}^{(k_1 k_2)}(\eta_1,\eta_2) & G_{22}^{(k_1 k_2)}(\eta_1,\eta_2)
\end{pmatrix},
\end{equation}

\begin{equation}
\begin{pmatrix}
\Gamma_{11}(\eta_2,\eta_3) & \Gamma_{12}(\eta_2,\eta_3)\\
\Gamma_{21}(\eta_2,\eta_3) & \Gamma_{22}(\eta_2,\eta_3)
\end{pmatrix}=\sum_{p_1,p_2=0}^{\infty}\begin{pmatrix}
\Gamma_{11}^{(p_2 p_3)}(\eta_2,\eta_3) & \Gamma_{12}^{(p_2 p_3)}(\eta_2,\eta_3)\\
\Gamma_{21}^{(p_2 p_3)}(\eta_2,\eta_3) & \Gamma_{22}^{(p_2 p_3)}(\eta_2,\eta_3)
\end{pmatrix},
\end{equation}
\begin{equation}
    G(\eta_3)=\sum_{l=0}^{\infty}G^{(l)}(\eta_3),
\end{equation}
with each linear operator $G_{ij}^{(k_1 k_2)}(\eta_1,\eta_2)$ homogeneous of degree $k_1$ in $\eta_1$ and of degree $k_2$ in $\eta_2$, for $i,j=1,2$, each linear operator $\Gamma_{ij}^{(p_2 p_3)}(\eta_2,\eta_3)$ homogeneous of degree $p_2$ in $\eta_2$ and of degree $p_3$ in $\eta_3$, for $i,j=1,2$ and each linear operator $G^{(l)}(\eta_3)$ homogeneous of degree $l$ in $\eta_3$, cf. \cite{CGK, CGS}.
Each of the operators $G_{ii}^{(k_1 k_2)}(\eta_1,\eta_2)$ with $i=1,2$ and $k_1,k_2\geq 0$, $\Gamma_{ii}^{(p_2 p_3)}(\eta_2,\eta_3)$ with $i=1,2$ and $p_1,p_2\geq 0$, and $G^{(l)}(\eta_3)$ with $l\geq 0$ is self-adjoint.
\begin{rem}
    The DN operators can be understood as certain pseudodifferential operators. More precisely, 
    let $m$ be a complex-valued function of real variable whose derivatives of any order have polynomial growth and setting $$D:=-i\partial_x$$ we define
    \begin{equation}
        (m(D)f)(x):=\frac{1}{2\pi}\int\int e^{ik(x-y)}m(k)f(y)dy\,dk.
    \end{equation}
    The operator $m(D)$ is called a Fourier multiplier operator and maps $\mathcal{S}(\mathbb{R})$ into
    $\mathcal{S}(\mathbb{R})$. Furthermore, 
    \begin{enumerate}
   \item $m(D)$ extends to a self-adjoint operator in $L^2(\mathbb{R})$ if and only if $m$ is real valued, cf. \cite{Reed}.
   \item $m(D)$ is bounded if and only if $m\in L^{\infty}(\mathbb{R})$
   \end{enumerate}
\end{rem}

Then, the leading order terms of the DN operators are given (cf. \cite{CGK}) by means of Fourier multipliers as
\begin{equation}
\begin{pmatrix}
G_{11}^0 & G_{12}^0\\
G_{21}^0 & G_{22}^0
\end{pmatrix}
= \begin{pmatrix}
D\coth(h_1 D) & -D\csch(h_1 D)\\
-D\csch(h_1 D)  & D\coth(h_1 D)
\end{pmatrix},
\end{equation}
 \begin{equation}\label{Eq:Gamma}
\begin{pmatrix}
\Gamma_{11}^0 & \Gamma_{12}^0\\
\Gamma_{21}^0 & \Gamma_{22}^0
\end{pmatrix}
=\begin{pmatrix}D\coth(h_3 D) & -D\csch(h_3 D)\\
-D\csch(h_3 D)  & D\coth(h_3 D)
\end{pmatrix}
\end{equation}
and 
\begin{equation}\label{Eq:G}
G^0=D\tanh(h_b D),
\end{equation} where $h_b:=h-h_3$ is the mean depth of the bottom layer, while $h_3$ is the mean depth of the middle layer.

Under consideration will be monochromatic solutions of the leading order linear equations, meaning that all components of these solutions are proportional to $\exp(ikx),$ where $k:=\frac{2\pi}{L}$ represents the wave number, $L$ being the corresponding wavelength. The latter means that each component is an eigenfunction of the operator $D$ with eigenvalue $k$. These facts allow us to replace the operator $D$ by its eigenvalue $k$ when acting on the previously mentioned monochromatic solutions. Consequently,
leading order approximations of the operators $A_{ij}(k)$ $(i,j=1,2,3)$, which are denoted by $A_{ij}^0(k)$ $(i,j=1,2,3)$, acquire the form

\begin{equation}\label{Aijk}
\begin{split}
&A_{11}^0(k)= \frac{k}{\rho_1}\coth(kh_1)-\frac{k\rho_2\Big( \rho_3\cosh(kh_3)\cosh(kh_b)+\rho_2 \sinh(kh_3)\sinh(k h_b) \Big)}{\rho_1 \sinh(kh_1) \Delta(k)}  ,\\
&A_{22}^0(k)= \frac{k}{\Delta(k)} \cosh(k h_1) \Big (\rho_2 \cosh(k h_3) \sinh(k h_b) + \rho_3\sinh(k h_3) \cosh(k h_b) \Big ),     \\
& A_{33}^0(k)= \frac{k}{\Delta(k)}\sinh(k h_b) \Big(\rho_1 \sinh(k h_1) \sinh(k h_3) + \rho_2 \cosh(k h_3) \cosh(k h_1)\Big),\\
&A_{12}^0(k)=  \frac{k}{\Delta(k)}\Big( \rho_2\cosh(kh_3)\sinh(kh_b)+\rho_3 \sinh(kh_3)\cosh(k h_b)    \Big) ,\\
&A_{13}^0(k)=  \frac{ k \rho_2 }{ \Delta(k)}   \sinh(k h_b),\\
&A_{23}^0(k)= \frac{k \rho_2 }{\Delta(k)}  \sinh(k h_b) \cosh(k h_1)   .
\end{split}
\end{equation}

where \begin{equation}\label{den}
\begin{split}
\Delta(k)&= \rho_2\Big(\rho_3 \cosh(k h_3) \cosh(k h_b)+\rho_2 \sinh(kh_3)\sinh(kh_b)\Big)\cosh(kh_1) \\
& +   \rho_1\Big(\rho_2 \cosh(k h_3) \sinh(k h_b)+\rho_3 \sinh(kh_3)\cosh(kh_b)\Big)\sinh(kh_1).  \\
\end{split}
\end{equation}
Note that 
\begin{equation}
\begin{split}
   \Delta(k)&>\rho_1^2\Big( \cosh(k h_3) \cosh(k h_b)+ \sinh(kh_3)\sinh(kh_b)\Big)\cosh(kh_1) \\
&+   \rho_1^2\Big(\cosh(k h_3) \sinh(k h_b)+ \sinh(kh_3)\cosh(kh_b)\Big)\sinh(kh_1) \\
&=\rho_1^2\Big(\cosh(kh_3+kh_b)\cosh(kh_1)+\sinh(kh_3+kh_b)\sinh(kh_1)\Big)\\
&=\rho_1^2\cosh(k(h_3+h_b+h_1))=\rho_1^2\cosh(k(h+h_1))>0.
\end{split}
\end{equation}
We observe that in the leading order all $A^0_{ij}(k)$ are even functions of $k$ and therefore the operators 
$A^0_{ij}(D)$ are self-conjugate operators. Thus, from \eqref{eta1}, \eqref{eta2} and \eqref{eta3} we have
\begin{equation}\label{alleta}
\begin{split}
\eta_{1,t}&=A^0_{11}(D)\xi_1+A^0_{12}(D)\xi_2+A^0_{13}(D)\xi_3,\\
\eta_{2,t}&=A_{12}^0 (D)\xi_1+A^0_{22}(D)\xi_2+A^0_{23}(D) \xi_3,\\
\eta_{3,t}&=A_{13}^{0}(D)\xi_1 +A_{23}^{0}(D)\xi_2+A^0_{33}(D)\xi_3.
\end{split}
\end{equation}
On the other hand, retaining only the linear terms in \eqref{E11}, \eqref{E22} and \eqref{E33} we obtain the equations
\begin{equation}\label{var_H_eta}
\begin{split}
\xi_{1,t}&=-g\rho_1\eta_1,\\
\xi_{2,t}&=-g(\rho_2-\rho_1)\eta_2\\
\xi_{3,t}&=-g(\rho_3-\rho_2)\eta_3.
\end{split}
\end{equation}
Denoting now $\mathfrak{m}=(\eta_1,\eta_2,\eta_3,\xi_1,\xi_2,\xi_3)^{T}$ we will solve the linear system of equations \eqref{alleta} - \eqref{var_H_eta} by utilising the Fourier transform:
setting $\hat{f}(k)=\frac{1}{\sqrt{2\pi}}\int_{\mathbb{R}}e^{-ikx}f(x)\,dx$ for $f\in\mathcal{S}(\mathbb{R})$, in each component of $\mathfrak{m}$, we see (after taking into account \eqref{Ham_syst},\eqref{HamGen} and \eqref{var_H_eta}) that the system \eqref{Ham_syst} is transformed
for any fixed $k\in\mathbb{R}$ into the linear autonomous system of ordinary differential equations
\begin{equation}\label{evol_syst}
\partial_t\hat{\mathfrak{m}}(k,t))=\mathcal{M}(k)\hat{\mathfrak{m}}(k,t)
\end{equation}
where 
\begin{equation}
\mathcal{M}(k)=\begin{pmatrix}
0 & 0 & 0 & A^0_{11}(k) & A^0_{12}(k) & A^0_{13}(k)\\
0 & 0 & 0 & A^0_{12}(k) & A^0_{22}(k) & A^0_{23}(k)\\
0 & 0 & 0 & A^0_{13}(k) & A^0_{23}(k) & A^0_{33}(k)\\
-g\rho_1 & 0 & 0 & 0 & 0 & 0\\
0 & g(\rho_1-\rho_2) & 0 & 0 & 0 & 0\\
0 & 0 & g(\rho_2-\rho_3) & 0 & 0 &0 
\end{pmatrix}.
\end{equation}

\begin{rem}
According to \cite{CICMP} the unique solution to  \eqref{evol_syst} is given by $\hat{\mathfrak{m}}(k,t)=e^{\mathcal{M}(k)t}\hat{\mathfrak{m}}_0(k)$
(where $\mathfrak{m}_0(k)=\frac{1}{\sqrt{2\pi}}\int_{\mathbb{R}}\mathfrak{m}_0(x)e^{ikx}\,dx$ is the initial data) and corresponds (via the inverse Fourier transform) to the solution
\begin{equation}\label{repres_sol}
    \mathfrak{m}(x,t)=\frac{1}{2\pi}\int_{\mathbb{R}}e^{\mathcal{M}(k)t}\hat{\mathfrak{m}}_0(k)e^{ikx}\,dk,\quad t\geq 0
\end{equation}
of \eqref{Ham_syst} with initial data $\mathfrak{m}_0\in\mathcal{S}(\mathbb{R}).$ We note now that, if $\lambda(k)=-ikc\,(c\neq 0)$ is a purely imaginary eigenvalue of $\mathcal{M}(k)$ with eigenvector $\mathfrak{v}(k)\neq 0$, then $e^{-ikc t}$ is an eigenvalue of $e^{\mathcal{M}(k)t}$ with eigenvector $\mathfrak{v}(k)$. Therefore, we see from \eqref{repres_sol}, that the solution $\mathfrak{m}(x,t)$ of \eqref{Ham_syst} acquires the form 
\begin{equation}
\frac{1}{2\pi}\int_{\mathbb{R}} e^{ik(x-ct)}\hat{\mathfrak{m}}_0(k)\,dk.
\end{equation}
This shows that a purely imaginary eigenvalue $\lambda(k)=-ikc$ of $\mathcal{M}(k)$ corresponds to the fundamental oscillation mode $e^{ik(x-ct)}\mathfrak{m}_0(k)$, with frequency $\frac{|k|}{2\pi}$, and constant speed $c.$
\end{rem}
\noindent Searching now for solutions $\hat{\mathfrak{m}}=\hat{\mathfrak{m}}_0 e^{ik(x-ct)}$ in \eqref{evol_syst} yields
$$(\mathcal{M}(k)+ikcI_6)\hat{\mathfrak{m}}=0,$$
equation which implies ${\rm det} (\mathcal{M}(k)+ikcI_6)=0$, since it is meaningful to ask $\hat{\mathfrak{m}}\neq 0$. Since for $\lambda\in\mathbb{R}$ it holds that (and accounting for $A^0_{ij}:=A^0_{ji}$ when $i>j$)
\begin{equation}\label{char_pol}
\begin{split}
{\rm det}&(\mathcal{M}(k)-\lambda I_6)\\
=&\lambda^6-(b_1 A^0_{11}+b_2 A^0_{22}+b_3 A^0_{33})\lambda^4\\
&+[b_1 b_2(A^0_{11} A^0_{22}-(A^0_{12})^2)+b_2 b_3 (A^0_{22} A^0_{33}-(A^0_{23})^2)+b_1 b_3(A^0_{11} A^0_{33}-(A^0_{13})^2)]\lambda^2\\
&-b_1 b_2 b_3{\rm det}(A^0_{i,j})_{i,j=\overline{1,3}},
\end{split}
\end{equation}
we have that the wave speed $c$ satisfies the equation
\begin{equation}\label{eq4c}
\begin{split}
c^6&+\frac{b_1 A^0_{11}+b_2 A^0_{22}+b_3 A^0_{33}}{k^2}c^4\\
&+\frac{b_1 b_2(A^0_{11} A^0_{22}-(A^0_{12})^2)+b_2 b_3 (A^0_{22} A^0_{33}-(A^0_{23})^2)+b_1 b_3(A^0_{11} A^0_{33}-(A^0_{13})^2)}{k^4}c^2\\
&+\frac{b_1 b_2 b_3{\rm det}(A^0_{ij})_{i,j=\overline{1,3}} }{k^6}=0,
\end{split}
\end{equation}
where $b_1=-g\rho_1,\,b_2=g(\rho_1-\rho_2),\,b_3=g(\rho_2-\rho_3).$
To ease the further calculations we find useful to introduce the notation
\begin{equation}\label{Deltas}
\begin{split}
&\Delta_1(k):=  \rho_2 \cosh(k h_3) \sinh(k h_b)+\rho_3 \sinh(kh_3)\cosh(kh_b)  ,\\
&\Delta_2(k):= \rho_2 \sinh(kh_3)\sinh(kh_b)+ \rho_3 \cosh(k h_3) \cosh(k h_b),\\
& \Delta_3(k):=\rho_2 \cosh(k h_3)\cosh(k h_1) + \rho_1 \sinh(k h_1)\sinh(k h_3),\\
&\Delta_4(k):=\rho_1 \cosh(k h_3)\sinh(k h_1) + \rho_2 \sinh(k h_3)\cosh(k h_1),\\
&\Delta_5(k):=\rho_1 \cosh(k h_1)\sinh(k h_3) + \rho_2 \sinh(k h_1)\cosh(k h_3),
\end{split}
\end{equation}
 which entails \begin{equation}\label{Del1}
\Delta(k)= \rho_1 \sinh(k h_1) \Delta_1(k) + \rho_2\cosh(kh_1) \Delta_2(k).
\end{equation}
Setting also
\begin{equation}\label{Del2}
\tilde{\Delta}(k)= \rho_3 \cosh(k h_b)\Delta_3(k)+\rho_2 \sinh(k h_b) \Delta_4(k),
\end{equation}
 we have
\begin{equation}\label{Minors}
\begin{split}
M_{12}:= &A^0_{11} A^0_{22}-(A^0_{12})^2=\frac{k^2  \tilde{\Delta}(k) \Delta_1(k)\sinh(k h_1)}{\rho_1\Delta^2(k)},\\
M_{13}:=  &A^0_{11} A^0_{33}-(A^0_{13})^2=  \frac{k^2 \tilde{\Delta}(k) \Delta_5(k)  \sinh(k h_b)}{\rho_1 \Delta^2(k)}  ,\\
M_{23}:= & A^0_{22} A^0_{33}-(A^0_{23})^2 =\frac{k^2  \tilde{\Delta}(k) \cosh(k h_1) \sinh (k h_b)\sinh(k h_3) }{\Delta^2(k)} .
\end{split}
\end{equation}
The coefficients of equation \eqref{eq4c} can be now written explicitly as

\begin{equation}\label{2layersAijk}
\begin{split}
a_4:=& \frac{b_1 A^0_{11}+b_2 A^0_{22}+b_3 A^0_{33}}{k^2}\\
=& \frac{g\Big(  (\rho_2-\rho_1)  (\rho_3-\rho_2)\sinh(k h_1) \sinh (k h_b) \sinh (k h_3) - \rho_2 \rho_3 \sinh(k(h_1+h_b+h_3))\Big)}{k \Delta(k)} ,\\
a_2:= &\frac{b_1 b_2 M_{12}+b_2 b_3 M_{23}+b_1 b_3 M_{13}}{k^4} \\
= &\frac{g^2 \tilde{\Delta}(k)}{k^2 \Delta^2(k)}\Big((\rho_2-\rho_1)(\rho_3-\rho_2)\cosh(k h_1) \sinh(k h_b) \sinh(k h_3) \Big.  \\ 
& \Big. \phantom{**************} +(\rho_2-\rho_1) \Delta_1 \sinh(k h_1) + (\rho_3-\rho_2) \Delta_5 \sinh(k h_b) \Big),\\
a_0: =&\frac{b_1 b_2 b_3{\rm det}(A^0_{ij})_{i,j=\overline{1,3}} }{k^6}= -\frac{g^3(\rho_2-\rho_1)(\rho_3-\rho_2)\sinh(k h_1)\sinh(k h_b) \sinh(k h_3)}{k^3 \Delta(k)}.
\end{split}
\end{equation}

\subsection{Short wave limit}
The short wave limit for the solutions of \eqref{eq4c} can be obtained by using the asymptotics $\sinh(k h_i) \sim \pm (1/2)\exp(kh_i)$ and $ \cosh (k h_i) \sim (1/2)\exp(k h_i) $ for $k\to \pm \infty.$  Thus 
\begin{align}
a_4 & \sim - \frac{g}{|k|}\left(1+ \frac{\rho_2 - \rho_1}{\rho_2 + \rho_1}+ \frac{\rho_3-\rho_2}{\rho_3+\rho_2} \right) , \\
    a_2 & \sim \frac{g^2}{k^2} \left(\frac{\rho_2 - \rho_1}{\rho_2 + \rho_1}+ \frac{\rho_3-\rho_2}{\rho_3+\rho_2} + \frac{\rho_2 - \rho_1}{\rho_2 + \rho_1} \cdot \frac{\rho_3-\rho_2}{\rho_3+\rho_2}  \right), \\
    a_0 & \sim - \frac{g^3}{|k|^3} \left(\frac{\rho_2 - \rho_1}{\rho_2 + \rho_1}\right) \left(\frac{\rho_3-\rho_2}{\rho_3+\rho_2}\right) .
\end{align} Thus, the roots of \eqref{eq4c} are real and given by the expressions
\begin{equation} \label{cshw}
   c_1^2 \sim  \frac{g}{|k|}  , \quad  c_{2}^2 \sim  \frac{g}{|k|} \left(\frac{\rho_2 - \rho_1}{\rho_2 + \rho_1}\right) , \quad c_{3}^2 \sim  \frac{g}{|k|} \left(\frac{\rho_3 - \rho_2}{\rho_3 + \rho_2}\right).
\end{equation}

\section{The long wave limit}\label{long}
\noindent This section is devoted to an analysis of the dispersion relation giving the wave speed $c$ in the long wave limit $k\to 0$ for the case of three-layer flows. We would like to note that the problem of several layers goes back to Stokes, and the general formula for the dispersion relation in the form of a determinant is provided in \cite{Gr1886} where it is attributed to a Tripos exam question in Cambridge (1884), where Stokes has been Lucasian Professor. While, the dispersion relation is a convoluted algebraic equation (whose degree is twice the number of layers), it exhibits interesting properties:
it was shown, \cite{Ben, Bai, Yih} in the setting of steady flows of  fixed mass flux, 
that all solutions of this equation ({which, in fact, give the wave speeds}) are real numbers.
The central topics of this section are bounds for the roots of the dispersion relation as well as some approximate formulas for these roots together with a comparison between the solutions of the approximate formulas and the numerical solutions of the dispersion relation \eqref{LWLim}.

\noindent Utilizing the findings in Section \ref{disp} we infer that the equation for $c$ corresponding to the long-wave limit ($k\to 0$) is 
\begin{equation}\label{LWLim}
\begin{split}
&c^6  -g(h_1+h_3+h_b)c^4\\[1em]
&+g^2\left[(\rho_2-\rho_1)\frac{h_1 h_3}{\rho_2}+(\rho_3-\rho_2)\frac{h_3 h_b }{\rho_3}
+(\rho_3-\rho_1)\frac{h_1 h_b }{\rho_3}\right]c^2\\[1em]
&-g^3(\rho_2-\rho_1)(\rho_3-\rho_2)\frac{h_1 h_3 h_b }{\rho_2\rho_3}=0,
\end{split}
\end{equation}
where we recall that the (mean) thickness of the upper layer is $h_1,$ of the middle layer is $h_3$ and of the bottom layer is $h_b.$ We also use the notation $h:=h_3+h_b$ and $H=h_1+h=h_1+h_3+h_b$ (the total average depth of the fluid).

\subsection{Bounds for the roots of the dispersion relation}
To study equation \eqref{LWLim} in greater detail we denote $c^2:=X$, and so we are led to consider the polynomial equation
\begin{equation}\label{P}
\begin{split}
&P(X):=X^3-g H X^2\\[1em]
&+g^2\left[\rho_1(\rho_2-\rho_1)\frac{h_1 h_3}{\rho_1\rho_2}+\rho_2(\rho_3-\rho_2)\frac{h_3 h_b }{\rho_2\rho_3}
+\rho_1(\rho_3-\rho_1)\frac{h_1 h_b }{\rho_1\rho_3}\right]X\\[1em]
&-g^3\rho_1(\rho_2-\rho_1)(\rho_3-\rho_2)\frac{h_1 h_3 h_b }{\rho_1\rho_2\rho_3}=0.
\end{split}
\end{equation}
The following result gives the interval where the roots of $P$ lie.
\begin{prop}
All the (real) roots of $P$ lie in the interval $(0, gH)$.
\end{prop}
\begin{proof}
We readily note from above that, the hypotheses $\rho_3>\rho_2>\rho_1$ and $h>h_3$, yield 
\begin{equation}\label{neg}
P(0)<0.
\end{equation}
Furthermore, it is also easy to see that if $H_m\geq H=h+h_1$ we have
\begin{equation}\label{pos}
\begin{split}
P(gH_m)=&g^3 H_m^2(H_m-H)+g^3 H_m\left[\rho_1(\rho_2-\rho_1)\frac{h_1 h_3}{\rho_1\rho_2}+\rho_1(\rho_3-\rho_1)\frac{h_1 h_b }{\rho_1\rho_3}\right]\\[0.5em]
&+\frac{g^3(\rho_3-\rho_2)h_3 h_b }{\rho_2\rho_3}\underbrace{(\rho_2 H_m- (\rho_2-\rho_1)h_1)}_{\geq\rho_2 h +\rho_1 h_1} >0.
\end{split}
\end{equation}
Inequalities \eqref{neg} and \eqref{pos} yield that there is $X_1\in (0, gH )$ such that $P(X_1)=0$. 
Close inspection of the coefficients of $P$ shows that $P(X)<0$ for all $X\leq 0$. This implies that all the real roots of $P$ lie in $(0, gH )$. 
\end{proof}
\noindent Although explicit formulas are available for the roots of $P$, the intricacy of these formulas makes them difficult to use. Instead, using a general result of Prodanov (\cite{Prod, Prod2}),
we provide bounds for the roots, defined in terms of the coefficients of $P$. First, we establish some notation.
Let 
\begin{align}
    \mathcal{A} & :=-gH  \label{defA} \\
      \mathcal{B} &:=g^2\left[\frac{\rho_2-\rho_1}{\rho_2}h_1 h_3+\frac{\rho_3-\rho_2}{\rho_3}h_3 h_b
+\frac{\rho_3-\rho_1}{\rho_3}h_1 h_b \right] \label{defB} \\
\mathcal{C}& :=-g^3(\rho_2-\rho_1)(\rho_3-\rho_2)\frac{h_1 h_3 h_b }{\rho_2\rho_3}, \label{defC}
\end{align}
be the coefficients of the polynomial from \eqref{P}. 

The location of the three roots of the polynomial $P(X) $ depends, cf.\cite{Prod, Prod2}, on the ratio $\mathcal{B}/\mathcal{A}^2,$ which will be analyzed in the sequel. For a fixed value of $\mathcal{A},$ (that is, for a fixed total depth $H$) the maximum of $\mathcal{B}$ can be found, for example, from the Lagrangian function
$$L(h_1,h_3,h_b,\lambda)=\mathcal{B}(h_1,h_3,h_b) - \lambda(h_1+h_3+h_b - H).$$
The maximal value of $\mathcal{B}$ is 
\begin{equation}
    \mathcal{B}_{max}= \frac{\rho_2 (\rho_3 -\rho_1) g^2 H^2}{\rho_2^2 + (3\rho_3-\rho_1)\rho_2 +\rho_3\rho_1}=
    \frac{\nu_2 (1-\nu_1 ) }{\nu_2^2 + (3-\nu_1)\nu_2 +\nu_1} \mathcal{A}^2,
\end{equation}
(where  $\nu_1=\frac{\rho_1}{\rho_3}<1,$ $\nu_2=\frac{\rho_2}{\rho_3}<1$) and is attained for
\begin{align}
   h_1^0&=  \frac{(\rho_2 + \rho_3)\rho_2}{\rho_2^2 + (3\rho_3-\rho_1)\rho_2 +\rho_3\rho_1 }H,\\
   h_3^0&= \frac{(\rho_3 - \rho_1)\rho_2}   {\rho_2^2 + (3\rho_3-\rho_1)\rho_2 +\rho_3\rho_1 }H,\\
   h_b^0&= \frac{(\rho_1 + \rho_2)  \rho_3}   {\rho_2^2 + (3\rho_3-\rho_1)\rho_2 +\rho_3\rho_1 }H.
\end{align}
It is easy to check that \begin{equation}\label{fract}
     \frac{\nu_2 (1-\nu_1 ) }{\nu_2^2 + (3-\nu_1)\nu_2 +\nu_1}<\frac{1}{3},\end{equation} thus we always have $ 0<\mathcal{B}< \frac{1}{3} \mathcal{A} ^2.$ In \cite{Prod, Prod2} three possible scenarios for the ratio
     $\frac{\mathcal{B}}{\mathcal{A} ^2}$ can occur and these are
    \begin{align}
0<\mathcal{B}< \frac{2}{9} \mathcal{A} ^2,  \label{case1}\\
\frac{2}{9} \mathcal{A} ^2  <\mathcal{B}< \frac{1}{4} \mathcal{A} ^2 , \label{case2} \\
\frac{1}{4} \mathcal{A} ^2 <\mathcal{B}< \frac{1}{3} \mathcal{A} ^2 . \label{case3}
    \end{align}
    In each of the above scenarios bounds detailing the location of the three positive roots of $P(X)$ are available, cf. \eqref{root_int} below. Before passing to a discussion on these bounds, we examine
    the term on the left-hand side of \eqref{fract} and so we identify some conditions for the fulfillment of \eqref{case1}, \eqref{case2}, \eqref{case3}. These conditions are given in terms of the ratio between the density of the top layer and suitable linear combinations of the densities of the other layers as follows
    \begin{align}
(i):\,\,0<\mathcal{B}< \frac{2}{9} \mathcal{A} ^2, \, \, \text{when} \, \, \nu_1& > \frac{\nu_2(3-2\nu_2)}{7\nu_2+2}, \, \, \text{i.e.} \, \,  \rho_1 > \rho_U=\frac{3\rho_3-2\rho_2}{7\rho_2+2\rho_3} \rho_2 ; \label{est1}\\
(ii):\,\,0  <\mathcal{B}< \frac{1}{4} \mathcal{A} ^2 
\, \, \text{when} \, \,  \nu_1 &> \frac{\nu_2(1-\nu_2)}{3\nu_2+1} \,\, \text{i.e.} \,\, \rho_1 >\rho_L= \frac{\rho_3-\rho_2}{3\rho_2+\rho_3} \rho_2  .
\label{est2} 
    \end{align} Of course, we always have $\rho_U>\rho_L.$ The above conditions in terms of the densities are only sufficient conditions. Unfortunately, there is no simple criterion (in terms of densities), which discriminates between \eqref{case1}, \eqref{case2}, and \eqref{case3}. For example, a necessary (but not sufficient) condition to have the situation in \eqref{case3} is $$\rho_1 < \rho_L= \frac{\rho_3-\rho_2}{3\rho_2+\rho_3} \rho_2 .$$
In most practical-relevant situations, the relative differences in densities are relatively small. It can be shown that this assumption leads to the situation in \eqref{case1}. More precisely, assuming that 
\begin{equation}\label{rel_dens_case1}
   \frac{\rho_2-\rho_1}{\rho_2}<\frac{4}{9},\quad \frac{\rho_3-\rho_2}{\rho_3}<\frac{4}{9},\quad  \frac{\rho_3-\rho_1}{\rho_3}<\frac{4}{9},
\end{equation}
and using the Cauchy-Schwartz inequality, we infer from \eqref{defB} that
\begin{equation}\label{B_upper}
    \mathcal{B}<\frac{4}{9} g^2(h_1 h+h_3(h-h_3))<\frac{4}{9} g^2(h_1 h+\frac{h^2}{4})<\frac{4}{9} g^2\frac{(h+h_1)^2}{2}=\frac{2}{9}\mathcal{A}^2.
\end{equation}
Therefore, condition \eqref{rel_dens_case1} is another sufficient condition that leads to the case \eqref{case1}.
In practical applications the internal waves occur between layers of one type of liquid (or very similar liquids) with slightly different properties (and densities) thus the condition \eqref{rel_dens_case1} covers a wide range of practically possible situations.  

According to \cite{Prod, Prod2}, when $ \mathcal{A}<0,$ for the three real roots of $P(X)$, denoted by $0<X_3\leq X_2\leq X_1$, in the case \eqref{case1} the following bounds are valid:
\begin{equation}\label{root_int}
    \begin{split}
    0<-\mathcal{C}\leq -\mathcal{C}_2, \\
        -\frac{\mathcal{C}}{\mathcal{B}}\leq X_3\leq\mu_2,\\
        \mu_2\leq X_2\leq\lambda_2,\\
        \lambda_1\leq X_1\leq \sigma_2
    \end{split}
\end{equation}
where the following notations are used:    
    \begin{align}
         \mathcal{C}_1& :=-\frac{2\mathcal{A}^3}{27}+\frac{\mathcal{A}\mathcal{B}}{3}+\frac{2}{27}\sqrt{(\mathcal{A}^2-3\mathcal{B})^3}, \quad  \mathcal{C}_2 :=-\frac{2\mathcal{A}^3}{27}+\frac{\mathcal{A}\mathcal{B}}{3}-\frac{2}{27}\sqrt{(\mathcal{A}^2-3\mathcal{B})^3}, \nonumber \\
\lambda_{1}&:=-\frac{\mathcal{A}}{2}+\sqrt{\frac{\mathcal{A}^2}{4}-\mathcal{B}},\quad \lambda_{2}:=-\frac{\mathcal{A}}{2} -\sqrt{\frac{\mathcal{A}^2}{4}-\mathcal{B}}, \nonumber \\
    \mu_1&:=-\frac{\mathcal{A}}{3}+\frac{\sqrt{3}}{3}\sqrt{\frac{\mathcal{A}^2}{3}-\mathcal{B}}, \quad \mu_2:=-\frac{\mathcal{A}}{3}-\frac{\sqrt{3}}{3}\sqrt{\frac{\mathcal{A}^2}{3}-\mathcal{B}},  \nonumber
    \end{align}
and
\begin{equation}
\sigma_1:=-\frac{\mathcal{A}}{3}-2\frac{\sqrt{3}}{3}\sqrt{\frac{\mathcal{A}^2}{3}-\mathcal{B}}, \quad
\sigma_2:=-\frac{\mathcal{A}}{3}+2\frac{\sqrt{3}}{3}\sqrt{\frac{\mathcal{A}^2}{3}-\mathcal{B}}. \nonumber
\end{equation}

\subsection{Approximate solutions for the wave speeds}
We provide here some further insights resulting from formula \eqref{LWLim}. More precisely, we take into account several propagation regimes defined by the mutual relations existing between monomials in \eqref{LWLim}. We distinguish the following cases:\\
(i) $a_2 c^2$ and $c_0$ are much smaller than $c^6$ and $a_4 c^4$.
Then the approximate equation we obtain has the root $\pm \sqrt{g H}.$ However, from estimates \eqref{root_int} for the largest root we observe that if $\mathcal{B}/\mathcal{A}^2\ll 1$ 
\begin{align}
    \sigma_2 &= |\mathcal{A}|\left[1- \frac{\mathcal{B}}{\mathcal{A}^2}+ \mathcal{O}\left( \left( \frac{\mathcal{B}}{\mathcal{A}^2}\right)^2\right) \right], \\
    \lambda_1 &= |\mathcal{A}|\left[1- \frac{\mathcal{B}}{\mathcal{A}^2}+  \mathcal{O}\left( \left( \frac{\mathcal{B}}{\mathcal{A}^2}\right)^2\right) \right] ,
\end{align}
and the next correction to the approximate expression $X_1\approx gH$ is
$$X_1\approx|\mathcal{A}|- \frac{\mathcal{B}}{|\mathcal{A}|}= gH-\frac{g}{H} \left[\frac{\rho_2-\rho_1}{\rho_2}h_1 h_3+\frac{\rho_3-\rho_2}{\rho_3}h_3 h_b
+\frac{\rho_3-\rho_1}{\rho_3}h_1 h_b \right]. $$ Thus, the propagation speed on the surface is approximately
\begin{equation} \label{a1}
    c_{1,a}=\pm \sqrt{gH-\frac{g}{H} \left[\frac{\rho_2-\rho_1}{\rho_2}h_1 h_3+\frac{\rho_3-\rho_2}{\rho_3}h_3 h_b
+\frac{\rho_3-\rho_1}{\rho_3}h_1 h_b \right]}.
\end{equation}

\begin{table}[]
\small
\begin{tabular}{|l||l|l|l||l|l|l||l|l|l||l|}
\hline
       No & $\rho_1$          &       $\rho_2 $           &    $\rho_3 $              &    $  h_1   $          &   $h_3 $         &      $h_b$             &     $ \frac{c_{1,a}}{\sqrt{gH}}$             &      $\frac{c_{2,a}}{\sqrt{gH}} $            &       $\frac{c_{3,a}}{\sqrt{gH}} $            &    $ \frac{c_i}{\sqrt{gH}}$                \\ 
        \hline
 \multirow{9}{*}{} & \multirow{9}{*}{} & \multirow{9}{*}{} & \multirow{9}{*}{} & \multirow{9}{*}{} & \multirow{9}{*}{} & \multirow{9}{*}{} & \multirow{9}{*}{} & \multirow{9}{*}{} & \multirow{9}{*}{} & \multirow{9}{*}{} \\
  1.  &    1000  &       1010     &  1020       &    200       &      100         &      300        &       0.9976    &     0.0680      &       0.0241            &           0.9976, \, 0.0630, \, 0.0261         \\
  2.  &       1000       &         1020      &     1100        &   10         &    100         &       3000     &            0.9987   &        0.0504      &       0.0075           &   0.9987, \, 0.0499, \, 0.0076                \\
  3. &    1005  &       1016     &  1027       &   600       &      100         &      3000        &       0.9984   &        0.0557    &        0.0115          &              0.9984, \, 0.0545 , \, 0.0118       \\
4. &    997  &       999     &  1000       &   6       &      2         &      20        &      0.9997   &       0.0233      &       0.00636           &              0.9997 , \, 0.0223,  \, 0.00664       \\
5. &    997  &       999     &  1000       &   10       &      20         &      50        &      0.9998   &       0.0212      &       0.00929         &              0.9998 ,  \, 0.0184  ,  \, 0.01078        \\
6. &    650  &       1500     &  5500       &   0.02       &      0.05         &     0.03        &      0.8839   &       0.4676      &      0.2378           &       0.8487 , \, 0.4346 , \, 0.3015        \\
7. &     620  &       1900     &  5900       &   0.02       &      0.05         &     0.03        &      0.8816   &       0.4720      &      0.2480           &       0.8458 , \, 0.4188 , \, 0.3304        \\
   
  \hline
\end{tabular}
\caption{ The approximate and the exact wave-speeds for some values of the parameters.
{The entries 1 and 2 correspond to internal waves, those in 3 refer to equatorial internal waves, 4 and 5 pertain to three-layer lakes, while the entries 6 and 7 refer to laboratory experiments.}
The dimensionalities of all quantities are given in the International System of Units (SI), $g=9.81$ m/s$^2$. All velocities are given in non-dimensional units by division of $\sqrt{gH}.$  The quantities $c_i$ are the positive roots of \eqref{LWLim}}. 
\label{demo-table}
\end{table}

(ii) The dominant terms are $a_4 c^4$ and $a_2 c^2.$ 
Indeed, from \eqref{root_int} for $X_2$ if $\mathcal{B}/\mathcal{A}^2\ll 1$ we have
\begin{align}
    \mu_2 &=  \frac{\mathcal{B}}{2|\mathcal{A}|}+ |\mathcal{A}|\mathcal{O}\left( \left( \frac{\mathcal{B}}{\mathcal{A}^2}\right)^2\right), \\
    \lambda_2 &= \frac{\mathcal{B}}{|\mathcal{A}|}+ |\mathcal{A}|\mathcal{O}\left( \left( \frac{\mathcal{B}}{\mathcal{A}^2}\right)^2\right),
\end{align}
and therefore $X_2$ is of order $\mathcal{B}/|\mathcal{A}|\ll X_1. $ Moreover, again from \eqref{root_int} $(|\mathcal{C}|/\mathcal{B})<\mu_2\approx  {\mathcal{B}}/{2|\mathcal{A}|},  $ thus $\mathcal{B}^2>2|\mathcal{A}\mathcal{C}|.$ This justifies the further assumption $\mathcal{B}^2 \gg |\mathcal{A}\mathcal{C}|,$ under which it can be seen that, 
when  $c^2=X_2$, the dominant terms in the equation are $a_4 c^4$ and $a_2 c^2.$ Then the solution of the approximate equation gives $ c_{2,a}^2=\mathcal{B}/|\mathcal{A}|$ or
\begin{equation}\label{a2}
    c_{2,a}=\pm \sqrt{\frac{g}{H} \left[\frac{\rho_2-\rho_1}{\rho_2}h_1 h_3+\frac{\rho_3-\rho_2}{\rho_3}h_3 h_b
+\frac{\rho_3-\rho_1}{\rho_3}h_1 h_b \right]   }.
\end{equation}
(iii) The dominant terms are $a_2 c^2$ and $a_0.$ From Vi\'ete formulas we have $$X_1X_2X_3=-\mathcal{C}=|\mathcal{C}|.$$  In the leading order $X_1=|A|$, $X_2=\mathcal{B}/|\mathcal{A}|$ therefore $X_3= |\mathcal{C}/\mathcal{B}|$ thus 
\begin{equation}\label{a3}
    c_{3,a}=\pm \sqrt{\frac{g(\rho_3 - \rho_2)(\rho_2-\rho_1)h_1 h_b h_3} {h_1h_b\rho_2(\rho_3-\rho_1) + h_b h_3\rho_2(\rho_3-\rho_2)+h_1h_3\rho_3(\rho_2-\rho_1)}}.  
\end{equation}
In these derivations we made the assumptions $\mathcal{B}/\mathcal{A}^2\ll 1$ and $\mathcal{B}^2 \gg |\mathcal{A}\mathcal{C}|,$ which are compatible with the inequalities of \eqref{case1}, and, of course with all kind of realistic physical data. Therefore, it is expected that the approximate formulas perform best in this case. Nevertheless, the results in Table \ref{demo-table} show that \eqref{a1}, \eqref{a2} and \eqref{a3} perform quite well also in \eqref{case2} and \eqref{case3}.

\begin{rem}
 The two-layer limit case $\rho_3 =\rho_2$, with depths $h_1$ and $h$ in \eqref{a1} leads to
\begin{equation} 
    c_{1,a}^*=\pm \sqrt{g(h_1+h)\left(1-\frac{ h_1 h}{(h_1+h)^2} \frac{\rho_2-\rho_1}{\rho_2}
 \right)},
\end{equation}
\end{rem} which appears also in \cite{CGK,Kodaira} etc.

\begin{rem}
    The two-layer limit case $\rho_3 =\rho_2$, with depths $h_1$ and $h$ in \eqref{a2} leads to the well known formula for the phase speed of internal wave, which is related to the so-called Brunt-V\"ais\"al\"a frequency,
    \begin{equation}\label{a22}
    c^*_{2,a}=\pm \sqrt{\frac{g  (\rho_2 - \rho_1) h h_1 }{\rho_2(h+h_1)}   }.
\end{equation}
As a matter of fact, the precise formula for two layers in the rigid lid approximation is formula (4.5) of \cite{CGK} but because in the ocean the two densities are very close, $\rho_1 \approx \rho_2,$ in many situations formula \eqref{a22} is used instead \cite{Mas,Kodaira}. 
\end{rem}

The propagation speed \eqref{a2} is associated with the disturbance described by $\eta_2,$ because the internal wave formed at this interface usually has a significant amplitude moving with the mentioned speed. Clearly, $\sqrt{gH}$ is the speed of the surface waves (over the total depth $H$) which due to the influence of the interfaces is modified to \eqref{a1}, while \eqref{a3} is the speed of internal waves, observable at the interface $y=\eta_3.$  In general, waves with all possible speeds have their traces on the surface and on both interfaces. However, their amplitudes can be of quite different magnitudes.  

Table \ref{demo-table} displays a comparison between the solutions of the approximate formulas above and the (positive) numerical solutions of \eqref{LWLim}. It is remarkable that the approximate formulae \eqref{a1} -- \eqref{a3} give a very good approximation (for a range of physically feasible values of the parameters) to the actual solutions of \eqref{LWLim}.  {The lines 1--6 in Table \ref{demo-table} correspond to \eqref{case1} and the approximation formulas provide very good approximations. 
The line 7 illustrates \eqref{case2}. It does not seem that there are any real-life examples satisfying case \eqref{case3}.
The cases \eqref{case2} and \eqref{case3} exhibit very large differences in the densities and are not realistic in an oceanographic context.} The approximation formulas work less well in the examples that exhibit large density contrasts, but nevertheless the approximate values are very good as a guide to the location of the exact ones. 

{The first example in Table 1 reflects typical parameter values for interfacial ocean internal waves. }
{The next example (shown as no. 2 in Table \ref{demo-table}) concerns oceanographic and limnological systems, where stratified water columns are
common due to variations in temperature, salinity, and pressure. A three-layer structure with densities of approximately 1000 $kg/m^3$, 1020 $kg/m^3$, and 1100 $kg/m^3$, and respective thicknesses of 10 m, 100 m, and 3000 m, can be justified by combining observations from coastal, open ocean, and hypersaline environments, cf. \cite{s_lay}, \cite{m_lay},  \cite{b_lay}.}

The third example (line 3 in Table \ref{demo-table})  is with parameter values, typical for the internal equatorial waves in the Pacific Ocean; see, for example, \cite{CICMP}. The total depth of the Pacific Ocean is approximately 4000 m on average, and the upper layer is several hundred meters deep. We note that standard seawater density rarely exceeds 1050 $kg/m^3$. 

{The next two examples (lines 4 and 5 in Table \ref{demo-table}) of a three-layer system concern temperate lakes, in which
particularly during the summer stratification period, the water column consistently organizes into three distinct layers: the epilimnion, metalimnion, and hypolimnion. The epilimnion, which typically extends from the surface to depths of 5 to 20 meters, is a warm, well-mixed layer influenced by wind and solar heating. Beneath it lies the metalimnion, or thermocline, a transition zone several meters thick, where temperature and density change rapidly with depth, effectively inhibiting vertical mixing. The hypolimnion, the coldest and densest layer, lies below the thermocline and remains largely isolated from surface influences, often becoming oxygen-depleted over time because of biological activity and lack of mixing. This three-layer structure is not just a convenient simplification, but a physically necessary framework that accurately captures the thermal, chemical, and biological dynamics of stratified lakes. Unlike oceanographic systems, where stratification is governed by both temperature and salinity over much greater depths, lake stratification is shallow, thermally driven, and seasonally reversible, reflecting the unique density behavior of freshwater and the limited depth of most inland water bodies, cf. \cite{Hut}.
}

{To investigate flow behavior in sharply stratified systems, we consider (in lines 6 and 7 of Table \ref{demo-table}) a three-layer fluid column with engineered density contrasts that far exceed those found in natural oceanographic environments \footnote{National Institute of Standards and Technology, US Department of Commerce, NIST Chemistry WebBook, SRD 69}. Unlike ocean stratification, which typically involves density differences of only 1–3 $\%$ due to temperature and salinity gradients, this system spans nearly an order of magnitude in density, using immiscible or semi-miscible fluids to maintain stable, well-defined interfaces.}

{The top layer typically consists of light hydrocarbons such as pentane or hexane. These fluids have densities in the range of 600–700 $kg/m^3$ and form a buoyant, sharply defined upper boundary. In experimental setups, this layer is often kept relatively thin, with a height of about 1–2 cm, to minimize mixing and maintain interface clarity.}

{The middle layer is composed of dense aqueous solutions, such as $ZnCl_2$ brines or glycerol-water mixtures, with densities ranging from 1,000 to 2,200 $kg/m^3$. This layer serves as a transitional buffer and is typically the thickest, with heights of 3–5 cm, allowing for stable stratification and tunable flow dynamics. Its density can be adjusted to fine-tune the separation between the top and bottom layers.}

{The bottom layer features ultra-dense fluids like perfluorodecalin or liquid metals such as gallium alloys, which can reach densities of 6,400 $kg/m^3$. This layer is usually 1–3 cm thick and provides a heavy, immiscible foundation that anchors the system and prevents upward mixing.
}

\section{The limit to two-layers with a free surface}  
\label{2layers}

We show here that the dispersion relation for the two-layer case, known from previous works \cite{CGK, Lamb}, follows from the considerations about the three-layer setting, as a particular case of equation \eqref{eq4c}.
In this case we set $\rho_3=\rho_2$. Then, taking into account \eqref{Aijk} and \eqref{den}, we have 
\begin{equation}\label{2layersAijk0}
\begin{split}
&A^0_{11} = k\frac{\tanh(kh) \coth(kh_1)+\frac{\rho_2}{\rho_1}}{\rho_2 \coth(k h_1) + \rho_1 \tanh(kh)},\\
&A^0_{12} =  k\frac{\tanh(kh) }{\sinh(kh_1)\big(\rho_2 \coth(k h_1) + \rho_1 \tanh(kh)\big)}  ,\\
&A^0_{22} = k\frac{\tanh(kh) \coth(kh_1) }{\rho_2 \coth(k h_1) + \rho_1 \tanh(kh)}.\\
\end{split}
\end{equation}
Therefore, expanding the determinant of the matrix $\mathcal{M}(k)-\lambda I_6$ by row three, we obtain
\begin{equation}\label{eq_lambda}
{\rm det}(\mathcal{M}(k)-\lambda I_6)=\lambda^2\Big[\lambda^4-(b_1 A^0_{11}+b_2 A^0_{22})\lambda^2+b_1 b_2\big(A^0_{11} A^0_{22}-(A^0_{12})^2\big)\Big].
\end{equation}
Further, we find from \eqref{2layersAijk0} that
\begin{subequations}\label{coeff}
\begin{align}
b_1 A^0_{11}+b_2 A^0_{22}=-g\rho_2 k\frac{1+\tanh(kh)\coth(kh_1)}{\rho_2\coth(kh_1)+\rho_1\tanh(kh)},\\[1em]
b_1 b_2\big(A^0_{11} A^0_{22}-(A^0_{12})^2)=\frac{g^2(\rho_2-\rho_1)k^2\tanh(kh_1)\tanh(kh)}{(\rho_2+\rho_1\tanh(kh_1)\tanh(kh))}.
\end{align}
\end{subequations}
Proceeding as before, that is searching for solutions $\hat{u}=\hat{u}_0 e^{ik(x-ct)}$ in \eqref{evol_syst} and denoting $\omega(k):=ck$, which is the spatial frequency of the wave solution, we notice from \eqref{eq_lambda} and \eqref{coeff} that $\omega$ satisfies the equation
\begin{equation}\label{eq_omega}
\omega^4-g\rho_2 k\frac{1+\tanh(kh)\coth(kh_1)}{\rho_2\coth(kh_1)+\rho_1\tanh(kh)}\omega^2+
\frac{g^2(\rho_2-\rho_1)k^2\tanh(kh_1)\tanh(kh)}{(\rho_2+\rho_1\tanh(kh_1)\tanh(kh))}=0,
\end{equation}
that is, we recover the equation obtained by Craig \emph{et al.}, cf. \cite{CGK} with Hamiltonian methods. This equation apparently goes back to \cite{Lamb} (1932) where the problem of "superimposed liquids" is solved.
\begin{rem}
Denoting $\omega^2:=Y$, we observe that $Y$ satisfies a second degree algebraic equation whose discriminant is 
$\Delta=(b_1 A_{11}-b_2 A_{22})^2+4b_1 b_2(A_{12})^2>0$, which shows that this equation has two real roots, called $Y_1,Y_2$.
Since
\begin{subequations}
\begin{align}
Y_1+Y_2=g\rho_2 k\frac{1+\tanh(kh)\coth(kh_1)}{\rho_2\coth(kh_1)+\rho_1\tanh(kh)}>0,\\[1em]
Y_1\cdot Y_2=\frac{g^2(\rho_2-\rho_1)k^2\tanh(kh_1)\tanh(kh)}{(\rho_2+\rho_1\tanh(kh_1)\tanh(kh))}>0,
\end{align}
\end{subequations}
we infer that $Y_1$ and $Y_2$ are positive, that is, equation \eqref{eq_omega} has four real roots.
\end{rem}

\section{The rigid lid model}\label{rigid}

In this section we discuss the rigid lid model, which corresponds to a flat surface $\eta_1(x,t)\equiv 0.$
The difference from the model with a nontrivial free surface is that for the domain $\Omega_1$ of the upper layer we have a single DN operator, namely $G_1(\eta_2).$ 

\subsection{DN operators and evolution equations}
We use the techniques developed in the previous sections and repeat the same steps. The Hamiltonian formulation can be established in this case as well, however we are going to employ only techniques related to the DN operators.
For the middle and the lower layers we have the DN operators \eqref{dn_m} and \eqref{dn_b}. The DN operator for the upper layer is like in (2.12) of \cite{CGK} and is defined as
\begin{equation}
    \left( \frac{\partial \varphi_1}{\partial (-\bf{n}_2)} \right)_{s_2}\sqrt{1+\eta_{2,x}^2}=-{\bf n_2}\cdot (\nabla \varphi_1)_{s_2}  \equiv - G_1(\eta_2) (\varphi_1)_{s_2}.
\end{equation}
More details are provided in formula (A.10) of \cite{CGK}, for example
\begin{equation}
    G_1(\eta_2)=\sum_{l=0}^{\infty}G_1^{(l)}(\eta_2),
\end{equation}

\begin{equation}\label{Eq:G1}
    G_1^{(0)}= D \tanh (h_1 D), \quad G_1^{(1)}=- D \eta_2(x) D + G_1^{(0)} \eta_2(x) G_1^{(0)} \ldots.
\end{equation}
Then, obviously
\begin{equation}\label{Eq:eta2}
    \eta_{2,t}=- G_1(\eta_2) (\varphi_1)_{s_2}=\Gamma_{11}(\varphi_2)_{s_2}+\Gamma_{12}(\varphi_2)_{s_3},
\end{equation} and, like in \eqref{GammaGij}, we obtain
\begin{equation}\label{Eq:eta3}
    \eta_{3,t}=G(\eta_3)(\varphi_3)_{s_3}=-[\Gamma_{21}(\varphi_2)_{s_2}+\Gamma_{22}(\varphi_2)_{s_3}].
\end{equation}
The equations \eqref{Eq:eta2} and \eqref{Eq:eta3} together with the 
\begin{equation} \label{eq:xi12}
    \xi_2:=\rho_2(\varphi_2)_{s_2}-\rho_1(\varphi_1)_{s_2}, \quad \xi_3:=\rho_3(\varphi_3)_{s_3}-\rho_2(\varphi_2)_{s_3}
\end{equation}
can be solved for $(\varphi_1)_{s_2},$ $(\varphi_2)_{s_2},$ $(\varphi_2)_{s_3}$ and $(\varphi_3)_{s_3}$ in terms of $\xi_2,$ $\xi_3$ and the corresponding DN operators (that depend on $\eta_2$ and $\eta_3$). We obtain the following expressions
\footnote{These expressions are different from those in \eqref{solQEq}, \eqref{phi13surf}.}:
\begin{align}
    (\varphi_2)_{s_2}&=B^{-1}\left[  \rho_1 G \xi_3 + (\rho_3 \Gamma_{22}+\rho_2 G) \Gamma_{12}^{-1} G_1 \xi_2 \right]=:\Phi^2_2(\eta_i, \xi_i) ,  \label{FFFa}\\
    (\varphi_2)_{s_3}&=-(B^*)^{-1}\left[  \rho_3 G_1 \xi_2 + (\rho_1 \Gamma_{11}+\rho_2 G_1) \Gamma_{21}^{-1} G \xi_3 \right]=:\Phi^2_3(\eta_i, \xi_i), \label{FFFb}
\end{align}
where
\begin{align}
    B= (\rho_3 \Gamma_{22}+\rho_2 G) \Gamma_{12}^{-1}(\rho_2 G_1+\rho_1 \Gamma_{11})-\rho_1\rho_3 \Gamma_{21}. \\
    \end{align}
This allows to express \eqref{Eq:eta2}, \eqref{Eq:eta3} in the form
\begin{align}\label{Eq:eta23}
    \eta_{2,t}&= A_{22}\xi_2+A_{23} \xi_3 \\
     \eta_{3,t}&= A_{32}\xi_2+A_{33} \xi_3 ,
\end{align}
where
\begin{align}
 A_{22}&=-\rho_3 \Gamma_{12} (B^*)^{-1} G_1+ \Gamma_{11} B^{-1}(\rho_3 \Gamma_{22}+\rho_2 G) \Gamma_{12}^{-1} G_1,  \label{A22}\\
    A_{23}&=\rho_1 \Gamma_{11} B^{-1} G -\Gamma_{12} (B^*)^{-1}(\rho_1 \Gamma_{11}+\rho_2 G_1) \Gamma_{21}^{-1} G, \\
    A_{32}&= \rho_3 \Gamma_{22} (B^*)^{-1} G_1 -\Gamma_{21} B^{-1}(\rho_3 \Gamma_{22}+\rho_2 G) \Gamma_{12}^{-1} G_1, \\
     A_{33}&=-\rho_1 \Gamma_{21} B^{-1} G +\Gamma_{22} (B^*)^{-1}(\rho_1 \Gamma_{11}+\rho_2 G_1) \Gamma_{21}^{-1} G  \label{A33}
\end{align} are operators that depend only on $\eta_2$ and $\eta_3.$ These operators, of course, are not the same as those in \eqref{Aij}. Thus, defining 
\begin{equation}\label{F23}
 F_2(\eta_i, \xi_i):=   A_{22}\xi_2+A_{23} \xi_3, \quad F_3(\eta_i, \xi_i):=A_{32}\xi_2+A_{33} \xi_3 , \, i=2,3,
\end{equation}
we have the nonlinear dynamics for $\eta_2$ and $\eta_3$ from equations \eqref{Eq:eta2} in the form
\begin{equation}\label{eta12tRL}
    \eta_{2,t}= F_2(\eta_i, \xi_i), \quad \eta_{3,t}=F_3(\eta_i, \xi_i),
\end{equation}
and also
\begin{align}\label{Eq:Phi23a}
     (\varphi_1)_{s_2}&=-G_1^{-1}\eta_{2,t}=-G_1(\eta_2)^{-1}F_2(\eta_i, \xi_i)=:\Phi^1_2(\eta_i, \xi_i), \\
      (\varphi_3)_{s_3}&=G^{-1}\eta_{3,t}=G(\eta_3)^{-1}F_3(\eta_i, \xi_i)=:\Phi^3(\eta_i, \xi_i). \label{Eq:Phi23b}
\end{align}
    The time-evolution of $\xi_2$ and $\xi_3$ is given by \eqref{E22}, \eqref{E33},
\begin{align}
        \xi_{2,t}&+\frac{\rho_2 (\Phi^2_{2,x})^2-\rho_1 (\Phi_{2,x}^1)^2-2F_2\xi_{2,x}\eta_{2,x}+(\rho_1-\rho_2)F_2^2   }{2(1+\eta_{2,x}^2)}+(\rho_2-\rho_1)g \eta_2 =0,  \label{Bern1}\\
        \xi_{3,t}&+\frac{\rho_3 (\Phi^3_{x})^2-\rho_2 (\Phi_{3,x}^2)^2-2F_3\xi_{3,x}\eta_{3,x}+(\rho_2-\rho_3)F_3^2   }{2(1+\eta_{3,x}^2)}+(\rho_3-\rho_2)g \eta_3 =0 \label{Bern2}
\end{align}
with the new definitions of $F_2,F_3$ in \eqref{F23} and the quantities in \eqref{FFFa}, \eqref{FFFb}, \eqref{Eq:Phi23a}, and \eqref{Eq:Phi23b}.

Thus we have the full nonlinear system of the time evolution in terms of DN operators given by \eqref{eta12tRL}, \eqref{Bern1}, \eqref{Bern2}. We will illustrate how this system produces model equations in different propagation regimes.

\subsection{Leading order linear approximation and dispersion relation}

In the leading order approximation we are using the DN operators from \eqref{Eq:G1}, \eqref{Eq:G}, and \eqref{Eq:Gamma} to obtain 
\begin{align}
 A_{22}^0(k)&= \frac{k[\rho_2 \cosh(kh_3)\sinh(kh_b) + \rho_3 \sinh(kh_3)\cosh(kh_b)]\sinh(k h_1)}{ {\bf D}(k) },\\
    A_{23}^0(k)&=A_{32}^0(k)=  \frac{k \rho_2 \sinh(k h_b)\sinh(k h_1) }{ {\bf D}(k) }, \\
     A_{33}^0(k)&= \frac{k[\rho_1 \cosh(k h_1) \sinh(kh_3) + \rho_2 \cosh(kh_3)\sinh(kh_1)]\sinh(k h_b)}{ {\bf D}(k) }, 
\end{align}
where \begin{equation}\label{den1}
\begin{split}
{\bf D}(k)&= \rho_2\Big(\rho_3 \cosh(k h_3) \cosh(k h_b)+\rho_2 \sinh(kh_3)\sinh(kh_b)\Big)\sinh(kh_1) \\
& +   \rho_1\Big(\rho_2 \cosh(k h_3) \sinh(k h_b)+\rho_3 \sinh(kh_3)\cosh(kh_b)\Big)\cosh(kh_1).  \\
\end{split}
\end{equation}
Note that for $k\ne 0,$ the sign of ${\bf D}(k)$ coincides with the sign of $k$ and thus  ${\bf D}(k)\ne 0.$ 
The equations for $\xi_i$ are like in \eqref{var_H_eta} and the linearised equations in the leading order approximation therefore are
\begin{align}
\eta_{2,t}&=A^0_{22}(D)\xi_2+A^0_{23}(D) \xi_3,\\
\eta_{3,t}&=A_{32}^{0}(D)\xi_2+A^0_{33}(D)\xi_3, \\
\xi_{2,t}&=-g(\rho_2-\rho_1)\eta_2,\\
\xi_{3,t}&=-g(\rho_3-\rho_2)\eta_3.
\end{align}
The analogue of the matrix $\mathcal{M}$ in \eqref{evol_syst} is now
\begin{equation}
\mathcal{M}(k)=\begin{pmatrix}
 0 & 0 &  A^0_{22}(k) & A^0_{23}(k)\\
 0 & 0 &  A^0_{23}(k) & A^0_{33}(k)\\
 g(\rho_1-\rho_2) &  0 & 0 &0 \\
 0 & g(\rho_2-\rho_3) & 0 & 0 
\end{pmatrix}.
\end{equation}
The characteristic equation ${\rm det} (\mathcal{M}(k)+ikcI_4)=0$ leads to the following bi-quadratic equation for the propagation speeds $c:$
\begin{align}
   &c^4   + a_2 c^2 +a_0 =0,  \label{eq4cRL}\\
   a_2&=-\frac{g(A^0_{22}(\rho_2 - \rho_1) + A^0_{33}(\rho_3-\rho_2)) }{k^2}=-\frac{g {\bf D}_1(k)}{k{\bf D}(k)},\\
   a_0&=  \frac{g^2(\rho_3 - \rho_2)(\rho_2 - \rho_1)(A^0_{22}A^0_{33} - (A^0_{23})^2)}{k^4} \nonumber \\
   & =     \frac{g^2(\rho_3 - \rho_2)(\rho_2 - \rho_1)\sinh(kh_1)\sinh(kh_3)\sinh(k h_b)}{k^2 {\bf D}(k)},
\end{align}
where \begin{align}
    {\bf D}_1(k)=& [\rho_2(\rho_3-\rho_1)\cosh(k h_3) \sinh(k h_b) + \rho_3(\rho_2-\rho_1) \sinh(kh_3)\cosh(kh_b)]\sinh(k h_1) \nonumber \\
    &+\rho_1(\rho_3-\rho_2)\cosh(kh_1) \sinh(kh_3)\sinh(k h_b).
\end{align}

\subsection{Short wave limit}

The short wave limit for the solutions of \eqref{eq4cRL} can be obtained in a similar way by the asymptotics $\sinh(k h_i) \sim \pm (1/2)\exp(kh_i)$ and $ \cosh (k h_i) \sim (1/2)\exp(k h_i) $ for $k\to \pm \infty.$  Thus 
\begin{align}
    a_2 & \sim -\frac{g}{|k|} \left(\frac{\rho_2 - \rho_1}{\rho_2 + \rho_1}+ \frac{\rho_3-\rho_2}{\rho_3+\rho_2} \right), \\
    a_0 & \sim \frac{g^2}{k^2} \left(\frac{\rho_2 - \rho_1}{\rho_2 + \rho_1}\right) \left(\frac{\rho_3-\rho_2}{\rho_3+\rho_2}\right) .
\end{align} The roots of \eqref{eq4cRL} are real as well, and coincide with the two of the already known values \eqref{cshw}
\begin{equation}
    c_{1}^2 \sim  \frac{g}{|k|} \left(\frac{\rho_2 - \rho_1}{\rho_2 + \rho_1}\right) , \quad c_{2}^2 \sim  \frac{g}{|k|} \left(\frac{\rho_3 - \rho_2}{\rho_3 + \rho_2}\right).
\end{equation}

\subsection{Long wave limit}

   Passing to the long-wave limit $k\to 0$ gives
   \begin{align}
       a_2 & \to -g\frac{\rho_2(\rho_3-\rho_1) h_1 h_b + \rho_3(\rho_2-\rho_1) h_1 h_3 + \rho_1 (\rho_3-\rho_2) h_3 h_b}{\rho_2 \rho_3 h_1 + \rho_1 \rho_2 h_b + \rho_1 \rho_3 h_3},\\
       a_0 & \to  \frac{g^2(\rho_3 - \rho_2)(\rho_2 - \rho_1) h_1 h_3 h_b}{\rho_2 \rho_3 h_1 + \rho_1 \rho_2 h_b + \rho_1 \rho_3 h_3},
   \end{align}
and, therefore, we obtain the equation
\begin{align}\label{fourth_order}
   (\rho_2 \rho_3 h_1 +& \rho_1 \rho_2 h_b + \rho_1 \rho_3 h_3) c^4 \nonumber \\
   &-g[\rho_2(\rho_3-\rho_1) h_1 h_b + \rho_3(\rho_2-\rho_1) h_1 h_3 + \rho_1 (\rho_3-\rho_2) h_3 h_b] c^2 \nonumber \\
   & +g^2(\rho_3 - \rho_2)(\rho_2 - \rho_1) h_1 h_3 h_b=0.
\end{align}
which appeared also in \cite{Rat,Mil} where it is obtained by other methods. We will prove in the following that equation \eqref{fourth_order} has four real roots for all values of the densities and all values of the (mean) depths $h_1,h_3,h_b$. To begin with, we denote $c^2=Z$ in \eqref{fourth_order} and obtain the second degree equation in $Z$ 
\begin{align}\label{Z_eq}
   (\rho_2 \rho_3 h_1 +& \rho_1 \rho_2 h_b + \rho_1 \rho_3 h_3) Z^2 \nonumber \\
   &-g[\rho_2(\rho_3-\rho_1) h_1 h_b + \rho_3(\rho_2-\rho_1) h_1 h_3 + \rho_1 (\rho_3-\rho_2) h_3 h_b] Z \nonumber \\
   & +g^2(\rho_3 - \rho_2)(\rho_2 - \rho_1) h_1 h_3 h_b=0.
\end{align}

Taking into account that $\rho_3>\rho_2>\rho_1$ and applying the Cauchy-Schwarz inequality we have first
\begin{align}
    &\rho_2(\rho_3-\rho_1) h_1 h_b + \rho_3(\rho_2-\rho_1) h_1 h_3 + \rho_1 (\rho_3-\rho_2) h_3 h_b \nonumber\\
    &\geq 2h_1\sqrt{h_3 h_b}\sqrt{\rho_2\rho_3(\rho_3-\rho_1)(\rho_2-\rho_1)}\nonumber \\
   & +2h_b\sqrt{h_1 h_3}\sqrt{\rho_1\rho_2(\rho_3-\rho_1)(\rho_3-\rho_2)}\nonumber\\
   &+2h_3\sqrt{h_1 h_b}\sqrt{\rho_1\rho_3 (\rho_2-\rho_1)(\rho_3-\rho_2)}\\
   &>2\sqrt{(\rho_3-\rho_2)(\rho_2-\rho_1)}\left[\sqrt{\rho_2\rho_3}h_1\sqrt{h_3 h_b}+\sqrt{\rho_1\rho_2}h_b\sqrt{h_1 h_3}+\sqrt{\rho_1\rho_3}h_3\sqrt{h_1 h_b} \right]\nonumber\\
   &=2\sqrt{(\rho_3-\rho_2)(\rho_2-\rho_1)}\sqrt{h_1 h_3 h_b}\left[\sqrt{\rho_2\rho_3 h_1}+\sqrt{\rho_1\rho_2 h_b}+\sqrt{\rho_1\rho_3 h_3}\right]\nonumber\\
   &\geq2\sqrt{(\rho_3-\rho_2)(\rho_2-\rho_1)}\sqrt{h_1 h_3 h_b} \sqrt{\rho_2\rho_3 h_1+\rho_1\rho_2 h_b+\rho_1\rho_3 h_3},\nonumber
\end{align}
where for the last inequality we have used that $\rho_i (i=1,2,3), h_1,h_3$ and $h_b$ are positive.
By squaring in the obtained inequality, we see that the discriminant of the equation \eqref{Z_eq} is strictly positive, that is, equation \eqref{Z_eq} has two distinct real roots, which, by the Vi\`ete's formulas, are positive. From the latter finding we infer that the fourth order equation \eqref{fourth_order} has four real roots: two positive and two negative.

An approximation for the smaller root we obtain by neglecting $c^4$ in \eqref{fourth_order},
\begin{equation}
    c_{-}:=\pm \sqrt{\frac{g(\rho_3 - \rho_2)(\rho_2 - \rho_1) h_1 h_3 h_b}{\rho_2(\rho_3-\rho_1) h_1 h_b + \rho_3(\rho_2-\rho_1) h_1 h_3 + \rho_1 (\rho_3-\rho_2) h_3 h_b }}
\end{equation} which, in turn, is approximately equal to the expression \eqref{a3} when $\rho_1 \approx \rho_2$ in the expression $ \rho_1 (\rho_3-\rho_2) h_3 h_b \approx  \rho_2 (\rho_3-\rho_2) h_3 h_b. $

By neglecting the last term in \eqref{fourth_order} we obtain an approximation for the larger root,
\begin{equation}
     c_{+}:=\pm \sqrt{\frac{g[\rho_2(\rho_3-\rho_1) h_1 h_b + \rho_3(\rho_2-\rho_1) h_1 h_3 + \rho_1 (\rho_3-\rho_2) h_3 h_b]}{\rho_2 \rho_3 h_1 + \rho_1 \rho_2 h_b + \rho_1 \rho_3 h_3}}.
\end{equation}
Under the further assumption that $\rho_1\approx \rho_2 \approx \rho_3 = \rho$ in the denominator and in the ratios $\rho_i/\rho_j$
\begin{align}
       c_{+} &\approx \pm \sqrt{\frac{g[\rho_2(\rho_3-\rho_1) h_1 h_b + \rho_3(\rho_2-\rho_1) h_1 h_3 + \rho_1 (\rho_3-\rho_2) h_3 h_b]}{\rho^2( h_1 + h_b +  h_3)}} \nonumber \\
        &  \approx \pm \sqrt{\frac{g}{h_1+h} \left(\frac{(\rho_2-\rho_1) h_1 h_3}{\rho_2}+ \frac{(\rho_3-\rho_2) h_3 h_b}{\rho_3}+\frac{(\rho_3-\rho_1) h_1 h_b}{\rho_3}   \right)}   ,  \nonumber 
\end{align} which is \eqref{a3}.

Therefore the approximation formulas for both models agree under the further assumption of small relative differences between densities, $|\rho_i - \rho_j|/\rho_n \ll 1.$

Finally, we point out that each propagation speed has its trace at every surface. For the largest propagation speed $c_+$ (mode 1) the ratio of the wave amplitudes of the upper $\eta_2$ and of that of the lower interface $\eta_3$ in the linear regime is
$  \left({\eta_2}/{\eta_3}\right)_{c=c_+}   >0, $ In this case both waves have the same polarity. 
For the slow-propagating mode $c_-,$ (mode 2), $  \left({\eta_2}/{\eta_3}\right)_{c=c_-}   <0,$ and in this case the waves have the opposite polarity.  See for example the details in \cite{Mil}.

\subsection{Boussinesq regime for the rigid lid system}\label{Boussinesq}
The next example of evolution equations arising from \eqref{eta12tRL}, \eqref{Bern1} and \eqref{Bern2} is the Boussinesq regime. As a result we obtain a coupled system of nonlinear evolution equations. We recall the expansions of the DN operators in the case under consideration \cite{CGK,IMT}:
\begin{align}
    G_1(\eta_2)&=D\tanh(h_1D)-D\eta_2 D+D\tanh(h_1D) \eta_2 D\tanh(h_1D)+\ldots , \label{G1}\\
     G(\eta_3)&=D\tanh(h_bD)+D\eta_3 D-D\tanh(h_b D) \eta_3 D\tanh(h_bD)+\ldots ,\\
     \Gamma_{11}(\eta_2, \eta_3)&= D \coth(h_3D) + D\csch(h_3D)\eta_3 D\csch (h_3D) \nonumber \\
     & \phantom{****}+ D\eta_2D-D\coth(h_3D)\eta_2D\coth(h_3D)+\ldots,  \\
     \Gamma_{12}(\eta_2,\eta_3)&= -D\csch(h_3D)-D\csch(h_3D)\eta_3 D\coth (h_3D)  \nonumber \\
      & \phantom{****}+D \coth(h_3D)\eta_2 D \csch(h_3D)+\ldots \\
      \Gamma_{21}(\eta_2,\eta_3)&=  -D\csch(h_3D) -D\coth(h_3D)\eta_3 D \csch(h_3D) \nonumber \\
       & \phantom{****}+D\csch(h_3D) \eta_2 D \coth(h_3D)+\ldots, \\
       \Gamma_{22}(\eta_2, \eta_3)&= D \coth(h_3D) + D\coth(h_3D)\eta_3 D\coth (h_3D) - D \eta_3 D  \nonumber \\
       & \phantom{****}-D\csch(h_3D) \eta_2 D \csch(h_3D)+\ldots.  \label{Gamma22}
     \end{align}
We use these expansions in the Boussinesq regime, that is, assuming that the small amplitude parameters for each combination $\varepsilon_{ij} =\frac{|\eta_j|_{\text{max}}}{h_i}$ are of the same order of magnitude, say $\varepsilon,$ for all $i,j$  and also the long-wave parameters $\delta_j=\frac{2 \pi h_j}{\lambda}=k h_j$ are of the same order $\delta,$ and  moreover $\varepsilon \ll1 $ and $ \delta^2 \ll 1$ are of the same order as well.  Furthermore, we recall the fact that the eigenvalue of $D$ is the wave-number $k$, so the action of $D$ will generate a quantity of order $\delta.$  On the other hand, the multiplication by $\eta_j$ will generate a quantity of order $\varepsilon.$  Keeping this in mind, the relevant expansion of \eqref{A22}--\eqref{A33} in terms of the DN operators \eqref{G1} -- \eqref{Gamma22} is:
\begin{align}
    A_{22}=A_{220}D^2 + A_{221}D^4+ B_{222} D\eta_2D+B_{223} D\eta_3 D +\ldots,\\
    A_{23}=A_{32}=A_{230}D^2 + A_{231}D^4+ B_{232} D\eta_2D+B_{233} D\eta_3 D +\ldots,\\
    A_{33}=A_{330}D^2 + A_{331}D^4+ B_{332} D\eta_2D+B_{333} D\eta_3 D +\ldots ,
\end{align}
where the expansion coefficients are \begin{align}
   A_{220}&= \frac{h_1(h_3\rho_3 + h_b\rho_2) }{ h_1\rho_2\rho_3+ h_3\rho_1 \rho_3 + h_b\rho_1\rho_2}, \\
    A_{221}&=- \frac{h_1^2 (h_1h_3^2\rho_1\rho_3^2 + 2h_1h_3h_b\rho_1\rho_2\rho_3 + h_1h_b^2\rho_1\rho_2^2 + h_3^3\rho_2\rho_3^2 + 3h_3^2h_b\rho_2^2\rho_3 + 3h_3h_b^2\rho_2^3 + h_b^3\rho_2^2\rho_3)}{3(h_1\rho_2\rho_3+ h_3\rho_1 \rho_3 + h_b\rho_1\rho_2)^2},\\
    B_{222}&= \frac{-(h_3\rho_3 + h_b\rho_2)^2\rho_1 + h_1^2\rho_2\rho_3^2}{(h_1\rho_2\rho_3+ h_3\rho_1 \rho_3 + h_b\rho_1\rho_2)^2},\\
     B_{223}&= \frac{(\rho_2-\rho_3)h_1^2\rho_2\rho_3 }{(h_1\rho_2\rho_3+ h_3\rho_1 \rho_3 + h_b\rho_1\rho_2)^2},
   \end{align}
\begin{align} 
     A_{230}&=\frac{ h_1 h_b \rho_2}{h_1\rho_2\rho_3+ h_3\rho_1 \rho_3 + h_b\rho_1\rho_2} ,\\
     A_{231}&= -  \frac{h_1 h_b \rho_2[6h_1h_bh_3 \rho_2^2+ h_b(2h_1^2 + 3 h_3^2)\rho_1\rho_2 + (3 h_3^2 + 2h_b^2)h_1\rho_2\rho_3 + h_3(2h_1^2 + h_3^2 +2 h_b^2)\rho_1\rho_3]}{6(h_1\rho_2\rho_3+ h_3\rho_1 \rho_3 + h_b\rho_1\rho_2)^2} , \\
      B_{232}&= -\frac{h_b \rho_1 \rho_2 [h_b \rho_2  +  (h_1 + h_3)\rho_3] }{( h_1\rho_2\rho_3+ h_3\rho_1 \rho_3 + h_b\rho_1\rho_2  )^2 } ,\\
       B_{233}&=\frac{h_1 \rho_3\rho_2 [(h_3 + h_b)\rho_1 + h_1 \rho_2 ]}{( h_1\rho_2\rho_3+ h_3\rho_1 \rho_3 + h_b\rho_1\rho_2)^2},
\end{align}
\begin{align}
     A_{330}&=\frac{h_b(h_1\rho_2 + h_3\rho_1)}{h_1\rho_2\rho_3+ h_3\rho_1 \rho_3 + h_b\rho_1\rho_2}\\
     A_{331}&=-\frac{h_b^2 (h_1^3\rho_1\rho_2^2 + 3h_1^2h_3\rho_2^3 + h_1^2h_b\rho_2^2\rho_3 + 3h_1h_3^2\rho_1\rho_2^2 + 2h_1h_3h_b\rho_1\rho_2\rho_3 + h_3^3\rho_1^2\rho_2 + h_3^2h_b\rho_1^2\rho_3)}{3(h_1\rho_2\rho_3+ h_3\rho_1 \rho_3 + h_b\rho_1\rho_2)^2},\\
      B_{332}&=\frac{h_b^2(\rho_1 -\rho_2)\rho_1\rho_2}{(h_1\rho_2\rho_3+ h_3\rho_1 \rho_3 + h_b\rho_1\rho_2)^2},\\
       B_{333}&=\frac{(h_1\rho_2 + h_3\rho_1)^2\rho_3 - h_b^2\rho_1^2\rho_2}{(h_1\rho_2\rho_3+ h_3\rho_1 \rho_3 + h_b\rho_1\rho_2)^2}.
\end{align}
Since in the Boussinesq regime the equations are maximum quadratic in the field variables, equations \eqref{Bern1}, \eqref{Bern2} suggest that we need only the leading order of the quantities $\Phi^i_j$ and $F_j:$ 
\begin{align}
  \Phi^2_2=&\frac{h_1\rho_3\xi_2 - h_b\rho_1\xi_3}{h_1\rho_2\rho_3+ h_3\rho_1 \rho_3 + h_b\rho_1\rho_2},\\
   \Phi^1_2=&-\frac{h_b \rho_2 (\xi_2 + \xi_3) + h_3\rho_3\xi_2}{h_1\rho_2\rho_3+ h_3\rho_1 \rho_3 + h_b\rho_1\rho_2},\\
   \Phi^3=&\frac{h_1 (\xi_2 +\xi_3)\rho_2 + h_3 \xi_3\rho_1}{h_1\rho_2\rho_3+ h_3\rho_1 \rho_3 + h_b\rho_1\rho_2} ,\\
   \Phi^2_3=&\frac{h_1\rho_3\xi_2 - h_b\rho_1\xi_3}{h_1\rho_2\rho_3+ h_3\rho_1 \rho_3 + h_b\rho_1\rho_2}.
\end{align}
Introducing the variables $\mathfrak{u}_j=(\xi_j)_x,$ we first notice, that $F_j$ in the leading order are linear combinations of $\mathfrak{u}_{j,x},$ while $\Phi^i_{j,x}$  are linear combinations of $\mathfrak{u}_{j}$ only. 
However, we notice that $\mathcal{O}(h_j \mathfrak{u}_{j,x}) = \mathcal{O}(kh_j \mathfrak{u}_{j}) = \delta \mathcal{O}(\mathfrak{u}_{j})\ll  \mathcal{O}(\mathfrak{u}_{j}).$  Thus the terms with $F_j$ in equations \eqref{Bern1}, \eqref{Bern2} can be neglected in comparison to the other terms with $\Phi^i_{j,x}$ (like in the water wave theory of a single layer). Then we obtain the system
\begin{equation}\label{coup_nonl_s}
\begin{split}
    \mathfrak{u}_{2,t}+ &\frac{[\rho_2(h_1\rho_3\mathfrak{u}_2 - h_b\rho_1\mathfrak{u}_3)^2-\rho_1(h_b \rho_2 (\mathfrak{u}_2 + \mathfrak{u}_3) + h_3\rho_3\mathfrak{u}_2)^2]_x}{2(h_1\rho_2\rho_3+ h_3\rho_1 \rho_3 + h_b\rho_1\rho_2)^2}   +(\rho_2-\rho_1)g\eta_{2,x}=0 ,\\
     \mathfrak{u}_{3,t}+ &\frac{[\rho_3(h_1 \rho_2 (\mathfrak{u}_2 +\mathfrak{u}_3) + h_3 \rho_1 \mathfrak{u}_3 )^2-\rho_2(h_1\rho_3\mathfrak{u}_2 - h_b\rho_1\mathfrak{u}_3)^2]_x}{2(h_1\rho_2\rho_3+ h_3\rho_1 \rho_3 + h_b\rho_1\rho_2)^2}   +(\rho_3-\rho_2)g\eta_{3,x}=0, \\
     \eta_{2,t}=&-A_{220}\mathfrak{u}_{2,x}  + A_{221}\mathfrak{u}_{2,xxx}- B_{222} [\eta_2 \mathfrak{u}_{2}]_x -B_{223} [\eta_3 \mathfrak{u}_{2} ]_x   \\
     &-A_{230} \mathfrak{u}_{3,x} + A_{231}\mathfrak{u}_{3,xxx}- B_{232} [\eta_2 \mathfrak{u}_{3}]_x -B_{233} [\eta_3 \mathfrak{u}_{3} ]_x   ,\\
\eta_{3,t}=&-A_{230}\mathfrak{u}_{2,x} + A_{231} \mathfrak{u}_{2,xxx} - B_{232} [\eta_2 \mathfrak{u}_{2}]_x 
-B_{233} [\eta_3  \mathfrak{u}_{2} ]_x  \\
& -A_{330}\mathfrak{u}_{3,x}  
+ A_{331} \mathfrak{u}_{3,xxx} - B_{332} [\eta_2 \mathfrak{u}_{3}]_x 
-B_{333} [\eta_3  \mathfrak{u}_{3} ]_x . 
\end{split}
\end{equation}
Letting $\rho_1 \to 0,$ we obtain
the two-layer domain with surface $y=\eta_2(x,t)$ and interface $y=\eta_3(x,t)$, whereby the coefficients $A_{ijk}, B_{ijk}, (i, j\in \{1,2\}, k\in\{0,1,2,3\}$) are given by
\begin{align}
   A_{220}& =\frac{h_3\rho_3 + h_b\rho_2}{ \rho_2\rho_3}, \quad A_{221} = - \frac{   h_3^3 \rho_3^2 + 3h_3^2h_b\rho_2 \rho_3 + 3h_3h_b^2\rho_2^2 + h_b^3\rho_2 \rho_3}{3\rho_2  \rho_3 ^2}, \nonumber \\
    B_{222}& =  \frac{1}{\rho_2  },\quad   B_{223} = \frac{\rho_2-\rho_3    }{\rho_2\rho_3 }; \nonumber
\end{align}
\begin{align} 
     A_{230}& = \frac{  h_b }{\rho_3 } ,\quad  A_{231} = -  \frac{ h_b [6h_bh_3 \rho_2+ (3 h_3^2 + 2h_b^2) \rho_3 ]}{6   \rho_3^2} , \quad    B_{232} =  0 ,\quad   B_{233} = \frac{ 1 }{\rho_3 }; \nonumber
\end{align}
\begin{align}
     A_{330}& = \frac{h_b  }{\rho_3}, \quad  A_{331} = -\frac{h_b^2 ( 3 h_3\rho_2  + h_b \rho_3 )}{3 \rho_3 ^2},\quad 
      B_{332} = 0,\quad   B_{333} =\frac{1}{\rho_3 }. \nonumber
\end{align}
Consequently, system \eqref{coup_nonl_s} becomes
\begin{equation}\label{recov_syst}
\begin{split}
\mathfrak{u}_{2,t}+ &\frac{1}{ \rho_2 } \mathfrak{u}_2 \mathfrak{u}_{2,x}  +\rho_2 g\eta_{2,x}=0 , \\
    \mathfrak{u}_{3,t}+ &\frac{\rho_2-\rho_3}{\rho_2 \rho_3 }\mathfrak{u}_2 \mathfrak{u}_{2 ,x} + \frac{1}{\rho_3}(\mathfrak{u}_2\mathfrak{u}_3)_x+ \frac{1}{\rho_3}\mathfrak{u}_3 \mathfrak{u}_{3 ,x}  +(\rho_3-\rho_2)g\eta_{3,x}=0 ,\\
     \eta_{2,t}=&- \frac{h_3\rho_3 + h_b\rho_2}{ \rho_2\rho_3} \mathfrak{u}_{2,x}  
      - \frac{   h_3^3 \rho_3^2 + 3h_3^2h_b\rho_2 \rho_3 + 3h_3h_b^2\rho_2^2 + h_b^3\rho_2 \rho_3}{3\rho_2  \rho_3 ^2} \mathfrak{u}_{2,xxx}\\
      &-\frac{1}{\rho_2} [\eta_2 \mathfrak{u}_{2}]_x - \frac{\rho_2-\rho_3    }{\rho_2\rho_3 } [\eta_3 \mathfrak{u}_{2} ]_x     \\
     &- \frac{  h_b }{\rho_3 } \mathfrak{u}_{3,x} -  \frac{ h_b [6h_bh_3 \rho_2+ (3 h_3^2 + 2h_b^2) \rho_3 ]}{6   \rho_3^2} \mathfrak{u}_{3,xxx} - \frac{ 1 }{\rho_3 } [\eta_3 \mathfrak{u}_{3} ]_x   ,\\
\eta_{3,t}=&- \frac{  h_b }{\rho_3 } \mathfrak{u}_{2,x}  -  \frac{ h_b [6h_bh_3 \rho_2+ (3 h_3^2 + 2h_b^2) \rho_3 ]}{6   \rho_3^2}  \mathfrak{u}_{2,xxx} -  \frac{ 1 }{\rho_3 } [\eta_3  \mathfrak{u}_{2} ]_x  \\
& -\frac{h_b  }{\rho_3} \mathfrak{u}_{3,x}  
-\frac{h_b^2 ( 3 h_3\rho_2  + h_b \rho_3 )}{3 \rho_3 ^2} \mathfrak{u}_{3,xxx} -  \frac{1}{\rho_3 }[\eta_3  \mathfrak{u}_{3} ]_x , 
\end{split}
\end{equation}
which recovers the system obtained by Craig et. al. \cite{CGK} on page 1626, derived by employing similar methods involving expansions of the DN operators for two layers with a free surface. 

The procedure utilized to obtain the Boussinesq system \eqref{recov_syst} highlights a general method for the derivation of nonlinear models for water waves propagation starting from the general system \eqref{eta12tRL}, \eqref{Bern1}, \eqref{Bern2}.

\subsection{Limit to the full equations for two layers with a free surface }

The limit to the full equations for two layers with a free surface as they appear (however, in Hamiltonian form) in \cite{CGK} could be done by simply setting up $\rho_1=0$ in all formulas in this Section. Furthermore, 
given a depth $h^*$, there is a relation between the DN operators $G_{ij}(h^*, \eta_l, \eta_u)$ and $\Gamma_{ij}(h^*, \eta_l, \eta_u)$, which reads 
\begin{align}
    G_{11}(h^*, \eta_l, \eta_u)& =\Gamma_{22}(h^*, \eta_l, \eta_u) , \quad   G_{12}(h^*, \eta_l, \eta_u)=\Gamma_{21}(h^*, \eta_l, \eta_u) , \label{GG1}\\
     G_{21}(h^*, \eta_l, \eta_u)& =\Gamma_{12}(h^*, \eta_l, \eta_u) , \quad   G_{22}(h^*, \eta_l, \eta_u)=\Gamma_{11}(h^*, \eta_l, \eta_u), \label{GG2}
\end{align}
where $y=\eta_l(x,t)$ is the lower boundary of the domain and $y=\eta_u(x,t)$ is the upper boundary of the domain

From \eqref{eq:xi12} we obtain 
\begin{equation}
     \xi_2:=\rho_2(\varphi_2)_{s_2} , \quad \xi_3:=\rho_3(\varphi_3)_{s_3}-\rho_2(\varphi_2)_{s_3}.
\end{equation}
We define the operator \footnote{This is essentially the operator $B$ in \cite{CGK} for the free surface case.} 
\begin{equation}
    \mathcal{B}:=\rho_3 \Gamma_{22}+\rho_2 G
\end{equation}
 which is self-adjoint. Then we have 
 \begin{align}
     B^{-1}= \frac{1}{\rho_2}G_1^{-1}\Gamma_{12} \mathcal{B} ^{-1 }, \quad  ( B^{-1})^*= \frac{1}{\rho_2} \mathcal{B} ^{-1 } \Gamma_{21}G_1^{-1 }.
 \end{align}
From \eqref{FFFa} and \eqref{FFFb} we obtain
\begin{align}
    (\varphi_2)_{s_2}&=\frac{1}{\rho_2} \xi_2 =:\Phi^2_2(\xi_2) ,  \label{FFFa1}\\
    (\varphi_2)_{s_3}&=-\mathcal{B} ^{-1 } \left[ G\xi_3 + \frac{\rho_3}{\rho_2}  \Gamma_{21} \xi_2 \right]=:\Phi^2_3(\eta_i, \xi_i). \label{FFFb1}
\end{align}
Then, from \eqref{Eq:eta3}, \eqref{FFFa1}, \eqref{FFFb1} after some calculations we have
\begin{equation} \label{Phi3RL}
    (\varphi_3)_{s_3}=-G^{-1}[\Gamma_{21}(\varphi_2)_{s_2}+\Gamma_{22}(\varphi_2)_{s_3}]=\mathcal{B} ^{-1 } \left[-\Gamma_{21}\xi_2+\Gamma_{22}\xi_3 \right]=:\Phi^3(\eta_i, \xi_i).
\end{equation}
In view of \eqref{GG1}-\eqref{GG2} these formulas reproduce (2.26)-(2.28) of \cite{CGK}. Next, from \eqref{A22} \eqref{A33} we obtain the expressions for the operators
\begin{align}
    A_{22}&= \frac{1}{\rho_2}\Gamma_{11}-\frac{\rho_3}{\rho_2}\Gamma_{12}\mathcal{B} ^{-1 }\Gamma_{21}, \\
    A_{23}&= -\Gamma_{12}\mathcal{B} ^{-1 }G, \\
   A_{32}&= \frac{1}{\rho_2}\left(-\mathcal{B}  +{\rho_3}\Gamma_{22}\right)\mathcal{B} ^{-1 }\Gamma_{21}=-G\mathcal{B} ^{-1 } \Gamma_{21}=A_{23}^* ,\\
    A_{33}&= \Gamma_{22}\mathcal{B} ^{-1 }G.
\end{align}
The operator $A_{33}$ is self-adjoint since $G$ and $\Gamma_{22}$ are self-adjoint operators, \cite{CGK}. Indeed, we can write it formally as 
$$A_{33}=\frac{1}{\rho_3} \Gamma_{22}\left( 1+\frac{\rho_2}{\rho_3} \Gamma_{22}^{-1} G \right)^{-1}\Gamma_{22}^{-1}G.$$
In the leading order $ \frac{\rho_2}{\rho_3}| \Gamma_{22}^{-1} G|= \frac{\rho_2}{\rho_3} |\tanh(kh_b) \tanh(k h_3)|<1, $ so, at least in some domain for the values of $(k,\eta_2,\eta_3)$ in the neighbourhood of $(0,0,0)$ we have an expansion 
$$A_{33}=\frac{1}{\rho_3} \left( G - \frac{\rho_2}{\rho_3}  G \Gamma_{22}^{-1 }G +   \frac{\rho_2^2}{\rho_3^2}  G \Gamma_{22}^{-1 }G  \Gamma_{22}^{-1 }G+...   \right)$$
from where it is evident that $A_{33}=A_{33}^*$ thus $A_{33}= G \mathcal{B} ^{-1 } \Gamma_{22}.$ The operator $A_{22}$ is also self-adjoint since $\Gamma_{11}$ is self adjoint and $\Gamma_{12}^*=\Gamma_{21},$ cf. \cite{CGK}.

The functionals $F_2$ and $F_3$ are as in \eqref{F23}, then the full equations of motion are the same as in \eqref{eta12tRL}, \eqref{Bern1} and \eqref{Bern2} (with $\rho_1=0$), that is, explicitly,
\begin{align}
    \eta_{2,t}&=\left( \frac{1}{\rho_2}\Gamma_{11}-\frac{\rho_3}{\rho_2}\Gamma_{12}\mathcal{B} ^{-1 }\Gamma_{21} \right) \xi_2 -\Gamma_{12}\mathcal{B} ^{-1 }G \xi_3 =: F_2(\eta_i,\xi_i) , \label{eta2RLfin} \\
    \eta_{3,t}&= G\mathcal{B} ^{-1 }[ -\Gamma_{21} \xi_2+ \Gamma_{22} \xi_3]=:F_3(\eta_i, \xi_i), \label{eta3RLfin}\\
     \xi_{2,t}&=-\frac{\rho_2 (\Phi^2_{2,x})^2-2F_2\xi_{2,x}\eta_{2,x}-\rho_2F_2^2   }{2(1+\eta_{2,x}^2)}-\rho_2g \eta_2 ,  \\
        \xi_{3,t}&=-\frac{\rho_3 (\Phi^3_{x})^2-\rho_2 (\Phi_{3,x}^2)^2-2F_3\xi_{3,x}\eta_{3,x}+(\rho_2-\rho_3)F_3^2   }{2(1+\eta_{3,x}^2)}-(\rho_3-\rho_2)g \eta_3 .  \label{Bern2RLfin}
\end{align}
The equation \eqref{eta2RLfin} follows also from \eqref{Eq:eta2}, \eqref{FFFa1}, \eqref{FFFb1}; the equation \eqref{eta3RLfin} follows also from \eqref{Eq:eta3} and \eqref{Phi3RL}.
The DN operators depend on $\eta_2$ and $\eta_3$ in a complicated nonlinear way, however, as we have mentioned, there is a systematic approach for their derivation. In the Boussinesq limit the above system produces the system \eqref{recov_syst}, see the derivation in \cite{CGK}, although we have obtained it as a limit from \eqref{coup_nonl_s}.  Fig. \ref{Fig2} shows the relations between the derived models in this Section. 
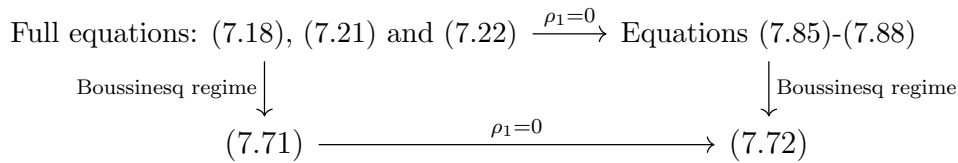
\begin{figure}
\[
\begin{tikzcd}
\text{\small{Full equations: 
\eqref{eta12tRL}, \eqref{Bern1} and \eqref{Bern2}}} \arrow{r}{\rho_1=0} \arrow[swap]{d}{\text{Boussinesq regime}} & \text{\small Equations \eqref{eta2RLfin}-\eqref{Bern2RLfin}} \arrow{d}{\text{Boussinesq regime}} \\
\eqref{coup_nonl_s} \arrow{r}{\rho_1=0} & \eqref{recov_syst}
\end{tikzcd}
\]
\caption{Relations between the derived models: limits and approximations.}
\label{Fig2}
\end{figure}

\section{Conclusions and Outlook}
\noindent We have presented a Hamiltonian formulation for the nonlinear governing equations describing irrotational two-dimensional inviscid and incompressible water flows with piece-wise constant density stratification in a three-layer fluid with a flat bottom, a free surface and two interfaces. We then introduced in each layer Dirichlet-Neumann operators with the help of which we re-expressed the Hamiltonian of the system and the equations of motion for the free surface and for the interfaces. The Hamiltonian formulation and the leading-order Dirichlet-Neumann operators have allowed us to obtain linearized equations for small-amplitude waves at the surface and interfaces.
We derived a bi-cubic equation representing the dispersion relation, in the short- and long-wave limits, we have analyzed the solutions of this equation and we obtained approximate expressions for the speed of the left- and right-running waves on the surface and at the internal interfaces. 

The results presented here open new avenues of investigations, of which we mention a few.
\begin{enumerate}
\item [ (i) ]Formulas \eqref{xiDef},\eqref{solQEq}, \eqref{phi13surf} give (through the Hamiltonian variables) the boundary values of the velocity potential which extends analytically to the interior of the fluid body. This will lead to in-depth investigations of important properties of the flow pertaining to the pressure, velocity field and particle trajectories, cf. e.g. \cite{DH}.

\item [ (ii) ] Since the three-layer scenario renders the inclusion of a current with piece-wise linear profile more realistic and natural (like for example in \cite{CIM,CICMP}), various results from the two-layer case concerning wave-current interactions in stratified flows are amenable to further extensions.   
In particular, the presence of currents in the layers will allow the formation of Kelvin-Helmholtz instabilities at the interfaces (see for example \cite{BenBri, BenBri2}). In this case, the instability condition at each interface will incorporate the current components in the three layers, displaying, therefore, an increased level of complexity to the problem which will require the pursuit of significant new research in this direction.

\item [ (iii) ] Furthermore, a detailed study of the Dirichlet-Neumann (DN) Operators will lead to weakly-nonlinear models for various propagation regimes which will be deduced from the equations presented in their general form, cf. \eqref{E11}, \eqref{E12}, \eqref{eta2}, \eqref{E22}, \eqref{eta3}, \eqref{E33} in the case of a free surface, and \eqref{eta12tRL}, \eqref{Bern1}, \eqref{Bern2} in the case of flat surface. An example was given by the derivation of the coupled system \eqref{recov_syst} in the Boussinesq regime. 
From the weakly nonlinear regimes like the KdV regime the soliton dynamics can be recovered, which is of practical importance given the fact that the solitons are stable and easier to observe in comparison to other wave types.    
The central role in the derivation of such systems is played by the DN operators and their expansions. Our motivation was based on the understanding that in the present research environment the use of symbolic computation will increase dramatically (perhaps powered by AI) and the expansions of the DN operators should be highly facilitated.
\end{enumerate}
The utilization of the Hamiltonian formulation and the DN operators allows a concise derivation of model equations for this type of stratified flows. We have demonstrated how these equations can be written in an explicit final form for all scales and propagation regimes. If necessary, the expansions of the DN operators lead to particular equations for a specific (small) scale parameter (up to the desired order of the small parameter).  For many configurations, especially the most common ones, these expansions are already known, which solves the problem immediately.  Highlighting the significance of the DN operators, we point out that recent years have seen a surge in the study of both the full nonlinear water wave problem and various approximate nonlinear water wave models. This progress is driven by a combination of methods focusing on DN operators and Hamiltonian systems, cf. \cite{Bal, Ber, Lan08, Did, CICMP, Lan_JAMS, Lan, IMT, Nach}, as well as methods for numerical computation of quantities, expressed by DN operators \cite{G17,GP,CGN}, following the pioneer work \cite{CS}.

\subsection*{Conflict of interest statement}

On behalf of all authors, the corresponding author states that there is no conflict of interest.

\subsection*{Data availability statement }

The reported results are of purely theoretical nature, and no data have been used supporting these results.

\subsection*{Acknowledgements} 
{The authors would like to thank three anonymous referees and the editor (Prof. J. Kirby) whose comments and suggestions greatly improved this paper.} R.I. gratefully acknowledges the financial support of Research Ireland under Grant number 21/FFP-A/9150. C.I.M. acknowledges the support from the project "Nonlinear Studies of Stratified Oceanic and Atmospheric Flows" funded by the European Union – NextGenerationEU and Romanian Government, under the National Recovery and Resilience Plan for Romania, contract no. 760040/23.05.2023, cod PNRR-C9-I8-CF 185/22.11.2022, through the Romanian Ministry of Research, Innovation, and Digitalization, within Component 9, Investment I8.

\end{document}